\documentclass[12pt,draftclsnofoot, peerreview, onecolumn]{IEEEtran}
\usepackage[T1]{fontenc}
\usepackage{gensymb}
\usepackage{balance}
\usepackage{color}
\usepackage{url}
\usepackage{bm}
\usepackage{cite}
\usepackage{algorithm}
\usepackage{threeparttable}
\usepackage{ifpdf}
\usepackage[dvips]{graphicx}
\graphicspath{{../}}
\DeclareGraphicsExtensions{.pdf,.jpeg,.png,.eps}
\ifCLASSOPTIONcompsoc
  \usepackage[tight,normalsize,sf,SF]{subfigure}
\else
  \usepackage[tight,footnotesize]{subfigure}
\fi
\usepackage{amsfonts}
\usepackage{amsmath}
\usepackage{mathrsfs}
\usepackage{multirow}
\usepackage{booktabs}
\usepackage{setspace}
\usepackage{amssymb}
\usepackage{mdwmath}
\usepackage{color}

\ifCLASSINFOpdf
\else
\fi

\newtheorem{example}{\textbf{Example}}
\newtheorem{definition}{\textbf{Definition}}
\newtheorem{theorem}{\textbf{Theorem}}
\newtheorem{lemma}{\textbf{Lemma}}
\newtheorem{corollary}{\textbf{Corollary}}
\newenvironment{proof}{\indent \textbf{\emph{Proof:}}}{\hfill $\blacksquare$}
\newenvironment{proofthm1}{\indent \textbf{\emph{Proof of Theorem 1:}}}{\hfill $\blacksquare$}
\newenvironment{counterex}{\indent \textbf{\emph{Counterexample:}}}{\hfill $\blacksquare$}
\newenvironment{example3}{\indent \emph{\textbf{Example}~3~(Cont'd):}}{}
\ifCLASSOPTIONonecolumn
\newcommand{\figwidth}{0.65\textwidth}
\fi

\ifCLASSOPTIONtwocolumn
\newcommand{\figwidth}{0.42\textwidth}
\fi

\begin{document}
\title{Systematic Convolutional Low Density Generator Matrix Code}

\author{
Suihua~Cai,
Wenchao~Lin,
Xinyuanmeng~Yao,
Baodian~Wei,
and~Xiao~Ma~\IEEEmembership{Member,~IEEE}
\thanks{$^*$Corresponding author is Xiao Ma. This work was supported by the NSF of China~(No. 61771499 and No. 61972431), the Science and Technology Planning Project of Guangdong Province~(2018B010114001), the National Key R\&D Program of China~(2017YFB0802503) and the Basic Research Project of Guangdong Provincial NSF~(No. 2016A030308008 and No. 2016A030313298).}
\thanks{This work was presented in part at 2018 IEEE Information Theory Workshop and 2016 IEEE International Symposium on Turbo Codes \& Iterative Information Processing.}
\thanks{
The authors are with the School of Data and Computer Science and Guangdong Key Laboratory of Information Security Technology, Sun Yat-sen University, Guangzhou 510006, China (e-mail:  caish5@mail2.sysu.edu.cn, linwch7@mail2.sysu.edu.cn, yaoxym@mail2.sysu.edu.cn,  weibd@mail.sysu.edu.cn, maxiao@mail.sysu.edu.cn).}
}


\maketitle

\begin{abstract}
  In this paper, we propose a systematic low density generator matrix~(LDGM) code ensemble, which is defined by the Bernoulli process.
  We prove that, under maximum likelihood~(ML) decoding, the proposed ensemble can achieve the capacity of binary-input output symmetric~(BIOS) memoryless channels in terms of bit error rate~(BER).
  The proof technique reveals a new mechanism, different from lowering down frame error rate~(FER), that the BER can be lowered down by assigning light codeword vectors to light information vectors.
  The finite length performance is analyzed by deriving an upper bound and a lower bound, both of which are shown to be tight in the high signal-to-noise ratio~(SNR) region.
  To improve the waterfall performance, we construct the systematic convolutional LDGM~(SC-LDGM) codes by a random splitting process.
  The SC-LDGM codes are easily configurable in the sense that any rational code rate can be realized without complex optimization.
  As a universal construction, the main advantage of the SC-LDGM codes is their near-capacity performance in the waterfall region and predictable performance in the error-floor region that can be lowered down to any target as required by increasing the density of the uncoupled LDGM codes. Numerical results are also provided to verify our analysis.
\end{abstract}



\section{Introduction}\label{SEC_1}

\IEEEPARstart{T}{he} channel coding theorem in~\cite{Shannon1948Theory} states that, as long as the transmission rate is below the channel capacity, there exists a coding scheme with infinite coding length for arbitrarily reliable transmission.
Shannon proved the channel coding theorem by analyzing the performance of the random code ensemble, which has no constraint on linearity.
In~\cite{Elias1955Coding}, it was proved that the totally random linear code ensemble can achieve the capacity of binary symmetric channels~(BSCs).
The same theorem was proved in~\cite{Gallager68} by deriving the error exponent.
In both cases, the random binary linear code ensemble is enlarged to the random coset code ensemble by adding a random binary sequence to each codeword.
Such an enlargement is a general technique to prove coding theorems for code ensembles without total randomness~\cite{Shulman1999RandomCoding}.
The above theorems imply the existence of capacity-achieving~(linear) code but do not give practical constructions of good codes, since a typical sample from the random~(linear) code ensemble has no efficient decoding algorithm even over binary erasure channels~(BECs).
For this reason, more attention has been paid to sparse linear codes, which can be decoded by the iterative belief propagation~(BP) algorithm.

The well-known low density parity-check~(LDPC) codes, which were proposed in~\cite{Gallager1962ldpc} and rediscovered in~\cite{MacKay1997LDPC}~\cite{Spielman1996LDPC}, are a class of sparse linear codes with sparse parity-check matrices.
With the help of density evolution~(DE) analysis~\cite{Urbanke2001Irregular}, which is developed to analyze the performance of LDPC code ensembles under iterative decoding, many capacity-approaching LDPC code ensembles have been designed~\cite{Urbanke2001LDPC45}.
Another class of sparse linear codes is the low density generator matrix~(LDGM) codes, which have sparse generator matrices.
Compared with the LDPC codes, the main issue of the LDGM codes is their non-negligible error floors, which, however, can be lowered down by concatenating outer codes.
For example, Raptor codes~\cite{Shokrollahi2006Raptor}, as concatenated codes with outer linear block codes and inner LT codes~\cite{Luby2002LT}, are proved to be capacity-achieving LDGM codes for BECs. In~\cite{Kakhaki2012Sparse}, an LDGM code ensemble with generator matrix defined by the Bernoulli process was introduced and proved to be capacity-achieving over BSCs.
In the existing proofs, generator matrices are typically of non-systematic form.
Hence, the code rate of the ensemble is slightly lower than the design rate.
To the best of our knowledge, no direct proof is available in the literature for systematic code ensembles.
The difficulty lies in the fact that the systematic generator matrices have the unity matrix as a non-random part.

Recently, the spatial coupling of LDPC codes~\cite{Kumar2014ThresholdSaturation} has revealed itself as a powerful technique to construct codes that achieve capacity universally over binary-input output-symmetric~(BIOS) memoryless channels.
The spatially coupled codes exhibit a threshold saturation phenomenon~\cite{Urbanke2011ThresholdSaturation}, which has attracted a lot of interest in the past few years.
The threshold saturation has been proved for BECs~\cite{Urbanke2011ThresholdSaturation} and generalized to BIOS memoryless channels~\cite{Kudekar2010}~\cite{Kudekar2013}.
The spatial coupling technique can also be applied to the LDGM codes.
In~\cite{Kumar2014ThresholdSaturation}, spatially coupled LDGM codes were also proved to achieve the capacity of BIOS channels.

As extension works of~\cite{Ma2016Coding}~\cite{Lin2018Coding}, we introduce systematic LDGM code ensembles in a different way.
For conventional LDGM/LDPC codes, the sparsity of the generator/parity-check matrices is characterized by the degree distribution.
In contrast, the proposed systematic LDGM code ensembles are defined according to a Bernoulli process with a small success probability. 
Different from the conventional capacity-achieving codes in terms of frame error rate~(FER), the proposed LDGM codes are proved to be capacity-achieving over BIOS memoryless channels in terms of bit error rate~(BER).
The proof technique developed in this paper shows that the BER can be lowered down by assigning light codeword vectors to light information vectors.
An upper bound and a lower bound on BER are derived to analyze the finite length performance of the proposed LDGM codes and are shown by numerical results to be tight in the high signal-to-noise ratio~(SNR) region. 
To reason the mismatching between the iterative BP decoding performance and the derived bounds, we carry out density evolution analysis over BECs for simplicity.
The DE results motivate us to employ spatial coupling techniques, by which the generator matrices become sparser and the edges in the decoding graph become ``roughly independent''~\cite{Urbanke2011ThresholdSaturation},
and propose the systematic convolutional LDGM~(SC-LDGM) codes.
The main advantage of the SC-LDGM codes is their easily predicable performance, leading to a \emph{universal} but \emph{simple} approach to constructing good codes of any rates.
Numerical results show that, under iterative BP decoding algorithm, the SC-LDGM codes perform about $0.7~{\rm dB}$ away from the Shannon limits for various code rates.

This paper is organized as follows.
In Section~\ref{SEC_2}, we introduce the systematic LDGM code ensemble and prove the coding theorem.
In Section~\ref{SEC_3}, we derive an upper bound and a lower bound on BER to analyze the finite length performance.
We also present density evolution analysis over BECs to analyze the performance of iterative BP decoding in the near-capacity region.
In Section~\ref{SEC_4}, we construct the SC-LDGM codes by a random splitting process to lower down the density of the generator matrices.
Numerical results show that the SC-LDGM codes have better performance in the waterfall region and match well with the analytical bounds in the error floor region.
Finally, some concluding remarks are given in Section~\ref{SEC_5}.

\section{Coding Theorem of Systematic LDGM Block Codes}\label{SEC_2}
\subsection{Systematic LDGM Block Codes}
Let $\mathbb{F}_2\stackrel{\Delta}{=}\{0, 1\}$ be the binary field.
A binary linear code $\mathscr{C}[n, k]$ with length $n$ and dimension $k$ is defined as a $k$-dimensional subspace of $\mathbb{F}_2^n$, which can be characterized by a generator matrix or a parity-check matrix.
A code {\em ensemble} is a collection of codes, each of which is assigned with a probability.
A convenient way to define a linear code ensemble is to generate randomly according to certain distributions either the generator matrix or the parity-check matrix.
Particularly, the following two code ensembles are of theoretical importance.
\begin{itemize}
  \item The totally random linear code ensemble $\mathscr{C}_h[n,k]$ can be characterized by a parity-check matrix $\mathbf{H}$ of size $(n-k)\times n$, where each element of $\mathbf{H}$ is drawn independently from a uniformly distributed binary random variable.
      The typical minimum distance of this code ensemble has been analyzed in~\cite{Gallager1962ldpc} as a benchmark for the LDPC code ensemble.
  \item The totally random linear code ensemble $\mathscr{C}_{g}[n,k]$ can be characterized by a generator matrix $\mathbf{G}$ of size $k\times n$, where each element of $\mathbf{G}$ is drawn independently from a uniformly distributed binary random variable.
      In~\cite{Gallager68}, the channel coding theorem has been proved by analyzing the performance of this code ensemble.
\end{itemize}

In a strict sense, the above two totally random linear code ensembles are different.
The code ensemble $\mathscr{C}_g[n,k]$ has some samples with code rates less than $k/n$, and the code ensemble $\mathscr{C}_h[n,k]$ has some samples with code rates greater than $k/n$.
Typically, a sample from $\mathscr{C}_{h}[n,k]$~(or $\mathscr{C}_{g}[n,k]$), which has generator matrices of high density and parity-check matrices of high density, has no efficient decoding algorithms even over BECs. A more practical code ensemble is the well-known LDPC code ensemble, which is first introduced in~\cite{Gallager1962ldpc}.
A sample from the LDPC code ensemble has a parity-check matrix of low density and~(hence) can be iteratively decoded. In~\cite{Urbanke2001Irregular}\cite{MacKay1999Irregular}\cite{Urbanke2001Irregular2}, regular and irregular LDPC code ensembles were defined by Tanner graphs with certain degree distributions.

In this paper, we consider the systematic code ensemble $\mathscr{C}_s[n,k]$ defined by the generator matrix $\mathbf{G} = [\mathbf{I}~\mathbf{P}]$, where $\mathbf{I}$ is the identity matrix of order $k$ and $\mathbf{P}$ is a random matrix of size $k\times (n-k)$.
Clearly, no matter what distribution of $\mathbf{P}$ is, the code rate is exactly $R=k/n$.
For theoretical analysis, we focus on the systematic biased random  code ensembles, which is defined as follows.
\begin{definition}
A linear block code ensemble is called a \emph{systematic biased random  code ensemble} if the generator matrix has the form $\mathbf{G} = [\mathbf{I}~\mathbf{P}]$ of size $k\times n$, where
\begin{equation}\label{DEF_LDGM_BLOCK_EN}
\mathbf{P} = \left(
 \begin{array}{ccccc}
   P_{0,0} & P_{0,1} & \cdots & P_{0,n-k-1}\\
   P_{1,0} & P_{1,1} & \cdots & P_{1,n-k-1}\\
   \vdots           & \vdots           & \vdots & \vdots            \\
   P_{k-1,0} & P_{k-1,1} & \cdots & P_{k-1,n-k-1}\\
 \end{array}
\right)
\end{equation}
and $P_{i, j}~(0\leqslant i\leqslant k-1, 0\leqslant j\leqslant n-k-1)$ is generated independently according to the Bernoulli distribution with success probability $\Pr\{P_{i,j} = 1\} = \rho\leqslant 1/2$.
\end{definition}

For decoding purposes, we are interested in the case that $\rho \ll 1/2$. With $\rho \ll 1/2$, this code ensemble has typical samples with generator matrices of low density and hence is termed as \emph{systematic LDGM code ensemble}. This ensemble can also be characterized by a random parity-check matrix $\mathbf{H} = [\mathbf{P}^{T}~\mathbf{I}]$, which is typically of low density with $\rho \ll 1/2$.
Therefore, a systematic LDGM code ensemble can also be viewed as a special class of LDPC code ensemble.
The speciality lies in the fact that the degree polynomials associated with the LDGM code ensemble have different meanings.
For example, the degree polynomial with respect to check nodes, when viewed as an LDPC code ensemble, is given by
\begin{equation}
M(x)=\sum_{i=1}^{k+1}M_ix^i=x(1-\rho+\rho x)^{k},
\end{equation}
where $M_i$ represents the \emph{probability} that a check node has degree $i$.
In contrast, if $M(x)$ is interpreted as the degree polynomial of a conventional LDPC code ensemble, the coefficient $M_i$ represents the \emph{fraction} of check nodes of degree $i$.
To see the difference, let us consider a sample code $\mathscr{C}[n,k]$.
If it is sampled from the conventional LDPC code ensemble, it will have exact $(n-k)M_i$ check nodes of degree $i$.
In contrast, if it is sampled from the systematic LDGM code ensemble, it may even have no check node of degree $i$.
The conventional LDPC code ensemble usually has a constant maximum degree and no variable node of degree one, while the proposed ensemble has nodes of degree one and an increasing maximum degree with the coding length.
As a result, the proposed ensemble has some samples with high check node~(variable node) degrees, but the probability assigned to such samples is negligible with large $n$ and small $\rho$.

\subsection{Coding Theorem}
Suppose that $\boldsymbol{u}=(u_0,\cdots,u_{k-1})$ of length $k$ is the data to be transmitted\footnote{For a vector $\boldsymbol{s} = (s_0,\cdots,s_{\ell-1})$, we use $\boldsymbol{s}_i^j$ to denote the subsequence $(s_i,\cdots,s_j)$ of $\boldsymbol{s}$. We also use $\boldsymbol{s}^{\ell}$ to emphasize the length of $\boldsymbol{s}$.}. The coded vector $\boldsymbol{x}=\boldsymbol{u}\mathbf{G}=(\boldsymbol{u},\boldsymbol{u}\mathbf{P})$ of length $n$, where $\boldsymbol{u}$ and $\boldsymbol{u}\mathbf{P}$ are referred to as the information vector and the parity-check vector, respectively, is transmitted over a noisy channel, resulting in a received sequence $\boldsymbol{y}$ of length $n$. Then the decoder employs a decoding algorithm  and outputs $\hat{\boldsymbol{u}}$ according to $\boldsymbol{y}$ as an estimation of $\boldsymbol{u}$.

In this paper, we focus on BIOS memoryless channels. A BIOS channel is characterized by an input set $\mathcal{X}=\mathbb{F}_2=\{0,1\}$, an output set $\mathcal{Y}$~(discrete or continuous), and a conditional probability mass~(or density) function $\{P_{Y|X}(y|x),x\in\mathbb{F}_2,y\in\mathcal{Y}\}$\footnote{In the case without causing much ambiguity, we omit the subscript of the probability mass~(or density) function in the remainder of this paper.}, which satisfies the symmetric condition that
$P_{Y|X}(y|1)=P_{Y|X}(\pi(y)|0)$ for some mapping $\pi: \mathcal{Y}\rightarrow\mathcal{Y}$ with $\pi(\pi(y))=y$ for all $y\in\mathcal{Y}$.
The channel~(used without feedback) is said to be memoryless if $P_{\boldsymbol{Y}|\boldsymbol{X}}(\boldsymbol{y}|\boldsymbol{x})=\prod_{t=0}^{n-1}P_{Y|X}(y_t|x_t)$. Let $P_X(1)=p$ and $P_X(0)=1-p$ be an input distribution of a BIOS memoryless channel.
The mutual information between the input and the output is given by
\begin{equation}
I(p)=(1-p)I_0(p)+pI_1(p),
\end{equation}
where
\begin{align}
I_0(p)&=\sum_{y\in\mathcal{Y}}P_{Y|X}(y|0)\log\frac{P_{Y|X}(y|0)}{P_Y(y)},\\
I_1(p)&=\sum_{y\in\mathcal{Y}}P_{Y|X}(y|1)\log\frac{P_{Y|X}(y|1)}{P_Y(y)},
\end{align}
and $P_Y(y)=(1-p)P_{Y|X}(y|0)+pP_{Y|X}(y|1)$.
For a BIOS memoryless channel, we have $I_0(p)=I_1(p)=\max_{0\leqslant p\leqslant 1} I(p)$ at $p=\frac{1}{2}$, which is the channel capacity.

Assume that the input vector to the encoder is uniformly distributed over $\mathbb{F}_2^k$. Let $E=\{\hat{\boldsymbol{U}}\neq \boldsymbol{U}\}$ be the event that the decoder output is not equal to the encoder input. Let $E_{i}=\{{\hat U_{i}}\neq U_i\}$ be the event that the $i$-th decoder output bit is not equal to the $i$-th encoder input bit. Obviously, we have $E = \bigcup_{i=0}^{k-1} E_{i}$. Then, we can define frame error rate as ${\rm FER}={\rm Pr}\{E\}$ and bit error rate as ${\rm BER}=\frac{1}{k}\sum_{1\leqslant i \leqslant k}{\rm Pr}\{E_{i}\}=\frac{1}{k}\textbf{E}[W_H(\hat{\bm{U}}+{\bm{U}})]$, where $\textbf{E}[\cdot]$ denotes the expectation of the random variable and $W_H(\cdot)$ denotes the Hamming weight function. In the remainder of this paper, the maximum likelihood~(ML) decoding algorithm is considered for FER and the maximum \emph{a posteriori}~(MAP) decoding is considered for BER, unless otherwise specified.

\begin{definition}
A sequence of codes~(code ensembles) $\mathscr{C}[n, k]$ are said to be capacity-achieving in terms of FER, if, for any $\epsilon > 0$, $\lim_{n \to \infty}k/n\geqslant C-\epsilon$ and $\lim_{n \to \infty}{\rm FER} = 0$, where $C$ is the channel capacity.
\end{definition}

\begin{definition}
A sequence of codes~(code ensembles) $\mathscr{C}[n, k]$ are said to be capacity-achieving in terms of BER, if, for any $\epsilon > 0$, $\lim_{n \to \infty}k/n\geqslant C-\epsilon$ and $\lim_{n \to \infty}{\rm BER} = 0$, where $C$ is the channel capacity.
\end{definition}

It is easy to see that capacity-achieving codes in terms of FER are also capacity-achieving in terms of BER. However, the converse is not true. This subtle difference can be shown by the counterexample below.

\begin{counterex}
Consider a sequence of codes $\mathscr{C}[n, k]$ with generator matrices $\mathbf{G}$ of size $k\times n$ over a BSC parameterized by the cross error probability $P_e$. Suppose that $\mathscr{C}[n, k]$ is capacity-achieving in terms of FER. It can be proved that $\mathscr{C}[n, k]$ is also capacity-achieving in terms of BER. However, for the sequence of codes $\mathscr{C}[n+1, k+1]$ defined by

\begin{equation}
\tilde{\mathbf{G}} = \left(
        \begin{array}{cc}
           1 & \mathbf{0}\\
           \mathbf{0} & \mathbf{G}\\
        \end{array}
\right),
\end{equation}
we have
\begin{equation}
{\rm BER}_{\tilde{\mathbf{G}}} =  \frac{P_{e}+k{\rm BER}_{\mathbf{G}}}{k+1} \rightarrow 0
\end{equation}
as $k \rightarrow \infty$, but
\begin{equation}
{\rm FER}_{\tilde{\mathbf{G}}}\geqslant P_e.
\end{equation}
\end{counterex}

It is well known that, as first proved by Elias~\cite{Elias1955Coding}, the totally random linear code ensemble $\mathscr{C}_{g}[n,k]$ can achieve the capacity of BSCs~(in terms of FER). In~\cite[{Theorem} 6.2.1]{Gallager68}, Gallager proved the same theorem by the use of the general coding theorem for discrete memoryless channels~(DMCs), which can be easily adapted to other BIOS channels.
The proof is for the coset codes by adding a random vector on the codewords and employs the pairwise independency between codewords.
Then the coding theorem for systematic random linear block codes is deduced as a corollary. These existing coding theorems imply that the systematic code ensemble $\mathscr{C}_s[n,k]$ with $\rho = 1/2$ is capacity-achieving in terms of FER~(also in terms of BER). In this paper, we prove the coding theorem for the systematic code ensemble $\mathscr{C}_{s}[n,k]$ with any given positive  $\rho \leqslant 1/2$.

\begin{theorem}\label{THEO_CAPACITY}
For any given positive $ \rho\leqslant 1/2$, the systematic code ensemble $\mathscr{C}_{s}[n,k]$ is capacity-achieving in terms of BER over BIOS memoryless channels.
\end{theorem}

The significance of the coding theorem proved in this paper includes the following aspects.
\begin{itemize}
  \item We give a direct proof for the coding theorem for the systematic code ensemble. To the best of our knowledge, no direct proof is available in the literature for the systematic code ensemble. The diffculty lies in the fact that the systematic generator matrices have the unity matrix as a non-random part. It is of interest to develop a direct proof, which may reveal more mechanism of good codes.
  \item Different from the coding theorem in~\cite{Elias1955Coding}~\cite{Gallager68} for the code ensemble $\mathscr{C}_{g}[n,k]$, which typically has no efficient decoding algorithm, we focus on the systematic code ensemble $\mathscr{C}_{s}[n,k]$ with $\rho \leqslant 1/2$. The simulation results show that, with $\rho \ll 1/2$, the LDGM code ensemble can be efficiently decoded by the iterative BP algorithm.
  \item Generally, the FER of the LDGM code ensemble is relative high because of the light codewords introduced by the sparse generator matrix. While the proof technique developed in this paper shows that systematic LDGM code ensemble is capacity-achieving in terms of BER, suggesting that the BER can be lowered down by assigning light codeword vectors to light information vectors.
\end{itemize}

\subsection{The Proof of Achievability }
Because of the linearity of the code, we assume that the all zero codeword $\boldsymbol{0}\in\mathbb{F}_2^n$ is transmitted over a BIOS memoryless channel, resulting in a received sequence $\boldsymbol{y}\in\mathcal{Y}^n$.
The maximum likelihood decoder selects $\boldsymbol{u}$ such that $P(\boldsymbol{y}|\boldsymbol{x})$ is maximized, where $\boldsymbol{x}$ is the codeword corresponding to $\boldsymbol{u}$.

To prove Theorem~\ref{THEO_CAPACITY}, we need the following two lemmas.

\begin{lemma}\label{LEM_T}
Over the systematic code ensemble $\mathscr{C}_{s}[n,k]$ defined by $\rho\leqslant 1/2$, the parity-check vector corresponding to an information vector with weight $w$ is a Bernoulli sequence with success probability \begin{equation}
\rho_{w}\triangleq \mathrm{Pr}\{X_j=1|W_H(\boldsymbol{U})=w\}=\frac{1-(1-2\rho)^w}{2}.
\end{equation}
Furthermore, for any given positive integer $T\leqslant k$,
\begin{equation}
P_G(\boldsymbol{x}_{k}^{n-1}|\boldsymbol{u})\triangleq\mathrm{Pr}\{\boldsymbol{X}_{k}^{n-1}=\boldsymbol{x}_{k}^{n-1}|\boldsymbol{U}=\boldsymbol{u}\} \leqslant P(\boldsymbol{0}^{n-k}|\boldsymbol{u})\leqslant (1-\rho_{T})^{n-k},
\end{equation}
for all $\boldsymbol{u} \in \mathbb{F}_{2}^{k}$ with $W_{H}(\boldsymbol{u})\geqslant T$ and $\boldsymbol{x}_{k}^{n-1} \in \mathbb{F}_{2}^{n-k}$.
\end{lemma}

\begin{proof}
By definition, the parity-check vector corresponding to an information vector $\boldsymbol{u}$ with weight $w \geqslant 1$ is $\boldsymbol{u}\mathbf{P}$. Since the elements of $\mathbf{P}$ are independent, identically distributed binary random variables, the success probability can be calculated recursively by $\rho_{1}=\rho$
and $\rho_{w+1} = \rho(1-\rho_w) + \rho_w(1-\rho)$. By induction, we can prove that
$\rho_{w}=[1-(1-2\rho)^w]/2$. Noticing that $\rho \leqslant \rho_{w} \leqslant \rho_{w+1}\leqslant 1/2$,  we have
$P(\boldsymbol{x}_{k}^{n-1}|\boldsymbol{u})\leqslant P(\boldsymbol{0}^{n-k}|\boldsymbol{u})\leqslant (1-\rho_{T})^{n-k} $ for all $\boldsymbol{u} \in \mathbb{F}_{2}^{k}$ with $W_{H}(u)\geqslant T$ and $\boldsymbol{x}_{k}^{n-1} \in \mathbb{F}_{2}^{n-k}$.
\end{proof}

\begin{lemma}\label{LEM_E0}
For a BIOS channel, the error exponent defined in~\cite[Theorem 5.6.2]{Gallager68} can be reduced as
\begin{equation}
E_r(R)=\max_{0\leqslant \gamma\leqslant 1}[E_0(\gamma)-\gamma R],\label{eq:Er}
\end{equation}
where
\begin{equation}
E_{0}(\gamma)=
-\log
\sum_{y\in\mathcal{Y}}P(y|0)^{1/(1+\gamma)}\left(\frac1{2}P(y|0)^{1/(1+\gamma)}+\frac1{2}P(y|1)^{1/(1+\gamma)}\right)^{\gamma}.
\label{eq:E0}
\end{equation}
Therefore, $E_r(R)>0$ for $R<I(1/2)$, where $I(1/2)$ is the BIOS channel capacity.
\end{lemma}

\begin{proof}
By symmetry, we see that the value of $E_{0}(\gamma)$ given in~(\ref{eq:E0}) remains unchanged if we interchange the labels of $0$ and $1$. That is,
\begin{equation}
E_{0}(\gamma)=
-\log
\sum_{y\in\mathcal{Y}}P(y|1)^{1/(1+\gamma)}\left(\frac1{2}P(y|0)^{1/(1+\gamma)}+\frac1{2}P(y|1)^{1/(1+\gamma)}\right)^{\gamma}.
\label{eq:E0'}
\end{equation}
Combining (\ref{eq:E0}) and (\ref{eq:E0'}), we have
\begin{align}
\notag
E_{0}(\gamma)&=
-\log
\sum_{y\in\mathcal{Y}}
\left(\frac1{2}P(y|0)^{1/(1+\gamma)}+\frac1{2}P(y|1)^{1/(1+\gamma)}\right)
\left(\frac1{2}P(y|0)^{1/(1+\gamma)}+\frac1{2}P(y|1)^{1/(1+\gamma)}\right)^{\gamma}\\
&=
-\log\sum_{y\in\mathcal{Y}}\left[\frac1{2}P(y|0)^{1/(1+\gamma)}+\frac1{2}P(y|1)^{1/(1+\gamma)}\right]^{1+\gamma}.
\end{align}
For BIOS memoryless channels, the random coding error exponent defined as~(5.6.16) in \cite[Theorem 5.6.2]{Gallager68} is maximized when $P_X(0)=P_X(1)=1/2$ and hence is reduced exactly the same as $E_r(R)$ given by~(\ref{eq:Er}).
\end{proof}

\begin{proofthm1}
From the law of total expectation, it follows that
\begin{align}
\notag \mathrm{BER}&=\sum_{\boldsymbol{y}\in\mathcal{Y}^n}P(\boldsymbol{y}|\bm{0})\cdot\mathrm{BER}|_{\boldsymbol{y}}.
\end{align}
As an upper bound of the  BER under the MAP decoding, we consider the BER under the ML decoding for the proof.
Given the received vector $\boldsymbol{y}$, the decoding output $\hat{\boldsymbol{U}}$ is a random vector over the code ensemble due to the randomness of the parity checks. Let $T\leqslant k$ be a positive integer. The event of decoding error can be split into two sub-events depending on whether or not $W_{H}(\hat{\boldsymbol{U}})\geqslant T$. Hence, the conditional BER can be upper bounded by
\begin{align}
\notag \mathrm{BER}|_{\boldsymbol{y}}&=\frac{\textbf{E}[W_H(\hat{\boldsymbol{U}})|\boldsymbol{y}]}{k}\\
\notag &= \sum_{\boldsymbol{u}}\mathrm{Pr}\{\boldsymbol{u}~\mathrm{is~the~most~likely}|\boldsymbol{y}\}\frac{W_H(\boldsymbol{u})}{k}\\
&\leqslant \frac{T}{k}+
\Big(\sum_{\boldsymbol{u}:W_H(\boldsymbol{u})\geqslant T}\mathrm{Pr}\{P(\boldsymbol{y}|\boldsymbol{u}\mathbf{G})\geqslant P(\boldsymbol{y}|\boldsymbol{0})\}\Big)^{\gamma},\textrm{~for~any~}0\leqslant \gamma \leqslant 1.\label{eq:condBER}
\end{align}
From the Markov inequality, for any given $s>0$, the probability of a vector $\boldsymbol{u}$  with $W_{H}(\boldsymbol{u})\geqslant T$ being more likely than $\boldsymbol{0}$ can be upper bounded by\\
\begin{align}
\notag
\mathrm{Pr}\{P({\boldsymbol{y}|\boldsymbol{u}\mathbf{G}}) \geqslant P({\boldsymbol{y}|\boldsymbol{0}})\}
\notag&\leqslant \frac{\textbf{E}[P^s({\boldsymbol{y}|\boldsymbol{u}\mathbf{G}})]}{P^s({\boldsymbol{y}|\boldsymbol{0}})}\\
\notag&=
\sum_{\boldsymbol{x}_{k}^{n-1}\in{\mathbb{F}_2^{n-k}}}
P_G(\boldsymbol{x}_{k}^{n-1}|\boldsymbol{u})
\frac{P^{s}(\boldsymbol{y}_{0}^{k-1}|\boldsymbol{u})P^{s}(\boldsymbol{y}_{k}^{n-1}|\boldsymbol{x}_{k}^{n-1})}{P^{s}(\boldsymbol{y}|\boldsymbol{0})}\\
&\notag\leqslant
\sum_{\boldsymbol{x}_{k}^{n-1}\in{\mathbb{F}_2^{n-k}}}
(1-\rho_{T})^{n-k}
\frac{P^{s}(\boldsymbol{y}_{0}^{k-1}|\boldsymbol{u})P^{s}(\boldsymbol{y}_{k}^{n-1}|\boldsymbol{x}_{k}^{n-1})}{P^{s}(\boldsymbol{y}|\boldsymbol{0})}\\
&=\left[\frac{1+(1-2\rho)^T}{2}\right]^{n-k}\frac{P^{s}(\boldsymbol{y}_{0}^{k-1}|\boldsymbol{u})}{P^{s}(\boldsymbol{y}_{0}^{k-1}|\boldsymbol{0}^k)}\sum_{\boldsymbol{x}_{k}^{n-1}\in{\mathbb{F}_2^{n-k}}}
\frac{P^{s}(\boldsymbol{y}_{k}^{n-1}|\boldsymbol{x}_{k}^{n-1})}{P^{s}(\boldsymbol{y}_{k}^{n-1}|\boldsymbol{0}^{n-k})},
\end{align}
where the second inequality follows from Lemma~\ref{LEM_T}.
Thus, we have
\begin{align}
\notag
&\sum_{{\boldsymbol{u}}:W_{H}({\boldsymbol{u}})\geqslant T}
{\rm Pr}\{P({\boldsymbol{y}|\boldsymbol{u}\mathbf{G}}) \geqslant P({\boldsymbol{y}|\boldsymbol{0}})\}\\
\notag\leqslant&
\sum_{{\boldsymbol{u}}:W_{H}({\boldsymbol{u}})\geqslant T}\left[\frac{1+(1-2\rho)^T}{2}\right]^{n-k}
\frac{P^{s}(\boldsymbol{y}_{0}^{k-1}|\boldsymbol{u})}{P^{s}(\boldsymbol{y}_{0}^{k-1}|\boldsymbol{0}^k)}
\sum_{\boldsymbol{x}_{k}^{n-1}\in{\mathbb{F}_2^{n-k}}}
\frac{P^{s}(\boldsymbol{y}_{k}^{n-1}|\boldsymbol{x}_{k}^{n-1})}{P^{s}(\boldsymbol{y}_{k}^{n-1}|\boldsymbol{0}^{n-k})}\\
\notag\leqslant&
\left[\frac{1+(1-2\rho)^T}{2}\right]^{n-k}\sum_{\boldsymbol{u}\in
	{\mathbb{F}_2^{k}}}
\frac{P^{s}(\boldsymbol{y}_{0}^{k-1}|\boldsymbol{u})}{P^{s}(\boldsymbol{y}_{0}^{k-1}|\boldsymbol{0}^k)}
\sum_{\boldsymbol{x}_{k}^{n-1}\in{\mathbb{F}_2^{n-k}}}
\frac{P^{s}(\boldsymbol{y}_{k}^{n-1}|\boldsymbol{x}_{k}^{n-1})}{P^{s}(\boldsymbol{y}_{k}^{n-1}|\boldsymbol{0}^{n-k})}\\
=&
\left[\frac{1+(1-2\rho)^T}{2}\right]^{n-k}\sum_{\boldsymbol{x} \in {\mathbb{F}_2^{n}}}\frac{P^{s}(\boldsymbol{y}_{0}^{n-1}|
	\boldsymbol{x})}{P^{s}(\boldsymbol{y}_{0}^{n-1}|\boldsymbol{0})}.
\end{align}
Substituting this bound into~(\ref{eq:condBER}), we have
\begin{align}
\notag
\mathrm{BER} &=
\sum_{\boldsymbol{y} \in {\mathcal{Y}}^n}P(\boldsymbol{y}|\bm{0})\cdot\mathrm{BER}|_{\boldsymbol{y}}\\
\notag&\leqslant
\sum_{\boldsymbol{y} \in {\mathcal{Y}}^n}
P(\boldsymbol{y}|\boldsymbol{0})\cdot
\left\{\frac{T}{k}+\left(\left[\frac{1+(1-2\rho)^T}{2}\right]^{n-k}\sum_{\boldsymbol{x} \in {\mathbb{F}_2^{n}}}\frac{P^{s}(\boldsymbol{y}_{0}^{n-1}|
	\boldsymbol{x})}{P^{s}(\boldsymbol{y}_{0}^{n-1}|\boldsymbol{0})}\right)^{\gamma}\right\}\\
\notag&\stackrel{(*)}=
\frac{T}{k}+
\left[1+(1-2\rho)^T\right]^{(n-k)\gamma}\cdot 2^{k\gamma}
\prod_{i=0}^{n-1}\left[\sum_{y_{i}\in \mathcal{Y}}P(y_{i}|0)
\left(\sum_{x_{i} \in \mathbb{F}_2}\frac1{2}\cdot\frac{P^{s}(y_{i}| x_{i})}{P^{s}(y_{i}|0)}\right)^{\gamma}\right]\\
\notag&=
\frac{T}{k}+
\left[1+(1-2\rho)^T\right]^{(n-k)\gamma}\cdot 2^{k\gamma}
\prod_{i=0}^{n-1}\left[\sum_{y_{i}\in \mathcal{Y}}P(y_{i}|0)^{1-\gamma s}
\left(\sum_{x_{i} \in \mathbb{F}_2}\frac1{2}P^{s}(y_{i}| x_{i})\right)^{\gamma}\right]\\
\notag&=
\frac{T}{k}+
2^{n\left\{\gamma\left[(1-R) \log(1+(1-2\rho)^T) + R\right]+\log\left[\sum_{y\in \mathcal{Y}}P(y|0)^{1-\gamma s}
	\left(\sum_{x\in \mathbb{F}_2}\frac1{2}P^{s}(y| x)\right)^{\gamma}\right]\right\}}\\
&\stackrel{(**)}=
\frac{T}{nR}+
2^{-n\left[E_0(\gamma)-\gamma\tilde{R}(T)\right]},\label{eq:BER_T}
\end{align}
where the equality $(*)$ follows from the memoryless channel assumption and the equality $(**)$ follows by setting $s = \frac1{1+\gamma}$ and denoting
\begin{equation}
\tilde{R}(T)\triangleq(1-R) \log(1+(1-2\rho)^T) + R.
\end{equation}
From~(\ref{eq:BER_T}), we see that the derived bound of BER is valid for any given positive integer $T\leqslant k$ and any $0\leqslant \gamma \leqslant 1$.
Then we can optimize  the bound as
\begin{equation}
\mathrm{BER} \leqslant \min_{T}\left\{\frac{T}{nR}+2^{-nE_r(\tilde{R}(T))}\right\}.\label{eq:BER_minT}
\end{equation}
Note that $\tilde{R}(T)$ converges to $R$ as $T\rightarrow \infty$.
Thus, for any given positive $\rho\leqslant 1/2$ and $0<R<I(1/2)$, there exists some $T_0>0$ such that $\tilde{R}(T_0)<I(1/2)$.
From Lemma~\ref{LEM_E0}, it follows that $E_r(\tilde{R}(T_0))>0$.
Hence, we see that both terms of the bound in~(\ref{eq:BER_minT}) converge to $0$ for sufficiently large $n$.

\end{proofthm1}

\section{Finite Length Performance of Systematic LDGM Block Codes}\label{SEC_3}
In Section~\ref{SEC_2}, we have proved the coding theorem by analyzing the BER performance of the LDGM code ensemble with infinite coding length. In practice, we are interested in the finite length performance of the LDGM code ensemble.
An upper bound and a lower bound on BER for a systematic linear code have been proposed in~\cite{Ma2017Systematic} for analyzing the performance of the systematic block Markov superposition transmission of repetition codes~(BMST-R) over AWGN channels.
In contrast to the BMST-R codes, which can also be viewed as a class of LDGM codes, the code ensemble defined in this paper~(Definition~\ref{DEF_LDGM_BLOCK_EN}) has an easily computable weight distribution as shown below.
We will also generalize these  bounds to the BIOS channels.

\subsection{Weight Distribution}
The input-redundancy weight enumerating function~(IRWEF) of a systematic block code is defined as~\cite{Benedetto1996IRWEF}
\begin{equation}
A(X,Y) = \sum_{i,j}A_{ij}X^iY^j,
\end{equation}
where $X, Y$ are two dummy variables and
$A_{ij}$ denotes the number of codewords having input~(information bits) weight $i$ and redundancy~(parity-check bits) weight $j$. For the systematic LDGM code ensemble, we have
\begin{equation}
A(X,Y) = 1 + \sum_{i=1}^{k}\binom{k}{i}X^i\left(1-\rho_{i}+\rho_{i}Y\right)^{n-k},
\end{equation}
where $\rho_{i}=[1-(1-2\rho)^i]/2$ as given in Lemma~\ref{LEM_T}.
This implies that the coefficients of the ensemble IRWEF can be given by
\begin{equation}
A_{ij} = \binom{k}{i}\binom{n-k}{j}\rho_i^j(1-\rho_i)^{n-k-j},
\end{equation}
for $1\leqslant i \leqslant k$ and $0\leqslant j \leqslant n-k$.

\subsection{Performance Bounds}
Suppose that the all zero codeword $\bm 0 \in \mathbb{F}_2^{n}$ is transmitted over a BIOS memoryless channel, resulting in a received sequence $\bm y \in \mathcal{Y}^n$.
Given a non-zero codeword $\bm c\in \mathbb{F}_2^{n}$, the pairwise error probability is defined conventionally as the probability that $\bm c$ is not less likely than $\bm 0$, which depends only  on the Hamming weight $W_H(\bm{c})$. Hence, we can denote $\Pr\{P(\bm y|\bm{c})\geqslant P(\bm y|\bm 0)\}$ as ${\rm PEP}(d)$, a function of $d=W_H(\bm{c})$.

\begin{theorem}\label{Bound}
For a BIOS memoryless channel, the BER of the systematic LDGM code ensemble $\mathscr{C}[n,k]$ under MAP decoding is upper bounded by
\begin{align}
\notag
{\rm BER_{MAP}}\leqslant \min_{0\leqslant r^*\leqslant k}\left\{\sum_{i=1}^{2r^*}\right.
&\frac{i}{k}\left( \sum_{j=0}^{n-k} A_{ij}{\rm PEP}(i+j) \right) \\
&\left.+\sum_{i=r^*+1}^{k}\frac{\min\{i+r^*,k\}}{k}\binom{k}{i}{\rm PEP}(1)^i(1-{\rm PEP}(1))^{k-i}\right\},
\end{align}
and lower bounded by
\begin{equation}
{\rm BER_{MAP}}\geqslant \sum_{w=0}^{n-k}P_W(w+1){\rm PEP}(w+1),
\end{equation}
where $P_W(w+1)$ is the probability that a row of the generator matrix has Hamming weight $w+1$, given by
\begin{equation}
P_W(w+1)=\binom{n-k}{w}\rho^w\left(1-\rho\right)^{n-k-w}.
\end{equation}
\end{theorem}

\begin{proof}
The proof is similar to those of Theorem 1 and Theorem 2 in~\cite{Ma2017Systematic}, and omitted here.
\end{proof}
%

Then, we have the following corollaries from Theorem~\ref{Bound}.
\begin{corollary}
For a BIOS memoryless channel, the BER of the systematic LDGM code ensemble $\mathscr{C}[n,k]$ under MAP decoding is upper bounded by
\begin{align}
\notag
{\rm BER_{MAP}}\leqslant \min_{0\leqslant r^*\leqslant k}\left\{\sum_{i=1}^{2r^*}\right.
&\frac{i}{k}\left( \sum_{j=0}^{n-k} A_{ij}z^{i+j} \right) \\
&\left.+\sum_{i=r^*+1}^{k}\frac{\min\{i+r^*,k\}}{k}\binom{k}{i}{\rm PEP}(1)^i(1-{\rm PEP}(1))^{k-i}\right\},
\end{align}
where $z = \sum_{y\in \mathcal{Y}}\sqrt{{P_{Y|X}(y|1)}{P_{Y|X}(y|0)}}$ is the Bhattacharyya parameter of the channel~\cite{Arikan2009ChannelPolarization}.
\end{corollary}

\begin{proof}
This can be proved by noting that ${\rm PEP}(d) \leqslant z^d$~\cite{Arikan2009ChannelPolarization}.
\end{proof}

\begin{corollary}
For an additive white Gaussian noise~(AWGN) channel with binary phase-shift keying~(BPSK) signalling, the BER of the systematic LDGM code ensemble $\mathscr{C}[n,k]$ under MAP decoding is upper bounded by
\begin{align}
\notag
{\rm BER_{MAP}}\leqslant \min_{0\leqslant r^*\leqslant k}\left\{\sum_{i=1}^{2r^*}\right.
&\frac{i}{k}\left( \sum_{j=0}^{n-k} A_{ij}Q\left(\frac{\sqrt{i+j}}{\sigma}\right) \right) \\
&\left.+\sum_{i=r^*+1}^{k}\frac{\min\{i+r^*,k\}}{k}\binom{k}{i}Q\left(\frac{1}{\sigma}\right)^i\left(1-Q\left(\frac{1}{\sigma}\right)\right)^{k-i}\right\},
\end{align}
\begin{equation}
{\rm BER_{MAP}}\geqslant \sum_{w=0}^{n-k}P_W(w+1)Q\left(\frac{\sqrt{w+1}}{\sigma}\right),
\end{equation}
where $\sigma^2$ is the variance of the noise and $Q(x)$ is the probability that a normalized Gaussian random variable takes a value not less than $x$.
\end{corollary}

\begin{proof}
This follows from the fact that ${\rm PEP}(d) = Q(\sqrt{d}/\sigma)$ for AWGN-BPSK channels.
\end{proof}

\begin{example}\label{EXEX_BC_k}
Consider the systematic LDGM code ensemble $\mathscr{C}_s[n,k]$ with $\rho = 0.01$ and rate $R=1/2$. The upper bounds and lower bounds for different coding lengths $n=512,1024,2048$ are shown in Fig.~\ref{FIG_BC_k}. The performance of uncoded transmission is also plotted. We can observe that
\begin{itemize}
  \item The upper bound and the lower bound match well in the high SNR region, implying that both the upper bound and the lower bound are tight in the high SNR region.
  \item The upper bound matches the performance curve of uncoded transmission in the low SNR region.
    This can be easily understood since a systematic code with direct transmission of information bits will perform no worse than the uncoded transmission in terms of SNR-BER curves.
  \item The performance predicted by the upper~(lower) bound improves with increasing coding length, implying that the MAP performance of the systematic LDGM codes improves with increasing coding length.
  Despite that the MAP decoding is typically infeasible, these bounds can be employed as a criterion to evaluate the optimality of a practical decoding algorithm.
\end{itemize}
\end{example}

\begin{figure}
  \centering
  \includegraphics[width=\figwidth]{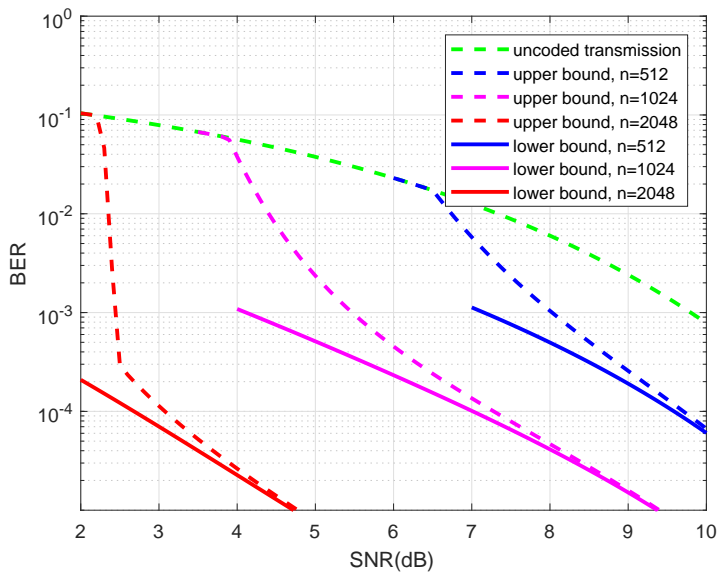}\\
  \caption{The upper bounds and lower bounds in Example~\ref{EXEX_BC_k}. We consider the systematic LDGM code ensemble $\mathscr{C}_s[n,k]$ with $\rho = 0.01$ and rate $R=1/2$. The coding lengths are $n=512,1024,2048$. The performance of uncoded transmission is also plotted.}\label{FIG_BC_k}
\end{figure}

\subsection{Decoding Algorithm}\label{SUBSEC_DEC_BC}
With small $\rho$ and large $n$, a sample from the systematic code ensemble $\mathscr{C}_s[n,k]$ typically has a generator matrix of low density and a parity-check matrix of low density, suggesting that the LDGM codes can be decoded by an iterative BP algorithm over the associated normal graphs~\cite{Forney2001Graph}. In the normal graph, edges represent variables and nodes represent constraints. Associated with each edge is a message that is defined in this paper as the probability mass function of the corresponding variable. All edges connected to a node must satisfy the specific constraint of the node. A full-edge connects to two nodes, while a half-edge connects to only one node. As shown in Fig.~\ref{FIG_BCDEC} , the normal graph of an LDGM code consists of the following two type of nodes.
\begin{itemize}
  \item \textbf{Node} \fbox{+}: It represents the constraint that the sum of all connecting variables must be zero over $\mathbb{F}_2$. The message updating rule at the node \fbox{+} is similar to that at the check node in an LDPC code. The only difference is that the messages of the half-edge need to be calculated from the channel observations.
  \item \textbf{Node} \fbox{=}: It represents the constraint that all connecting variables must take the same value. The message updating rule at the node \fbox{=} is the same as that at the variable node in an LDPC code.
\end{itemize}
In each iteration, the \fbox{=} is first updated and the \fbox{+} is subsequently updated. The decoding algorithm for an LDGM code is described in Algorithm~\ref{ALG_DEC_BC}.
\begin{figure}
  \centering
  \includegraphics[width=\figwidth]{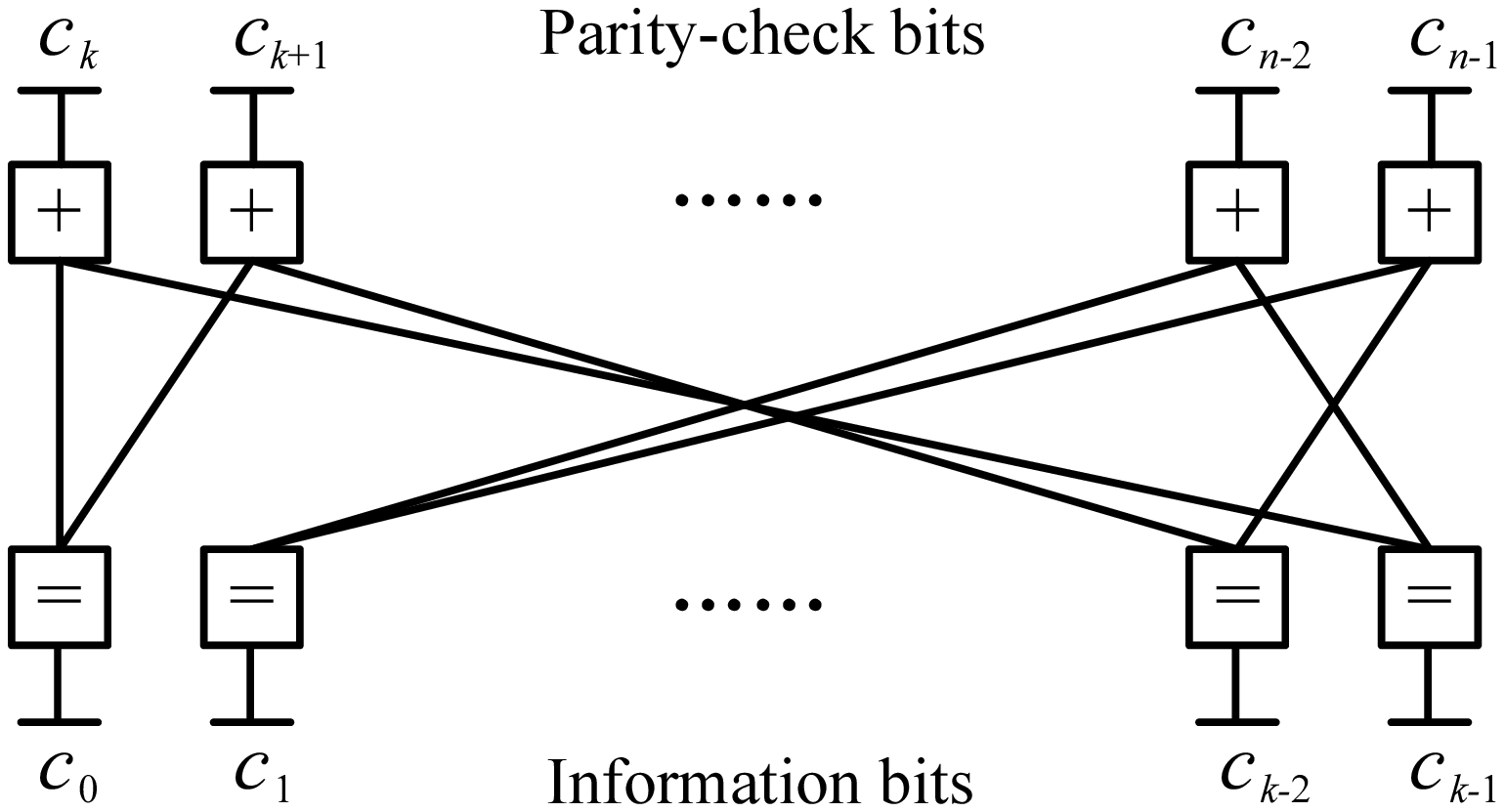}\\
  \caption{The normal graph of a systematic LDGM code $\mathscr{C}[n,k]$.}\label{FIG_BCDEC}
\end{figure}
\begin{algorithm}\caption{Iterative Decoding of the LDGM}\label{ALG_DEC_BC}
\begin{itemize}
  \item \textbf{Initialization}: Set a maximum iteration number $I_{\rm max}>0$. All messages over the half-edges are initialized by computing the \emph{a posteriori} probabilities with only the channel constraint. All messages over the full-edges are initialized as uniformly distributed variables.
  \item \textbf{Iteration}: For $i = 1, 2, \cdots , I_{\rm max}$,
  \begin{itemize}
    \item Update all the nodes of type \fbox{=}.
    \item Update all the nodes of type \fbox{+}.
  \end{itemize}
  \item \textbf{Decision}: Make decision on $\boldsymbol{u}$ by combining the soft extrinsic messages from \fbox{+} to \fbox{=} and the channel observations associated with $\boldsymbol{u}$, resulting in $\hat{\boldsymbol{u}}$.
\end{itemize}
\end{algorithm}

\begin{example}\label{EXEX_BC_fix}
We consider the systematic LDGM code ensemble $\mathscr{C}_s[2048,1024]$ with $\rho=0.010$ and $\rho=0.012$ transmitting over BPSK-AWGN channels.
The maximum iteration for decoding is set as $I_{\rm max} = 50$.
The BER performance with the corresponding upper bound and lower bound are shown in Fig.~\ref{FIG_BC_fix}, from which we can observe that
\begin{itemize}
  \item In the high SNR region, the simulated BER performance curves match very well with the respective theoretical~(lower) bounds, indicating that the iterative decoding algorithm is near optimal~(with respect to the MAP decoding algorithm).
  \item For a fixed coding length, the error floor can be lowered down by increasing $\rho$. However, under the sub-optimal iterative decoding, the performance in the low SNR region with a large $\rho$ is typically worse than that with a small $\rho$.
  \item In the low SNR region, the simulated BER performance curves are not predicted well by the derived bounds and are even worse than the upper bound for large $\rho$, which indicates that the iterative BP decoding is far from optimal for high-density codes. This also motivates us to carry out density evolution analysis and to employ spatial coupling techniques for lowering down the density of the generator matrices.

\end{itemize}
\end{example}
%
\begin{figure}
  \centering
  \includegraphics[width=\figwidth]{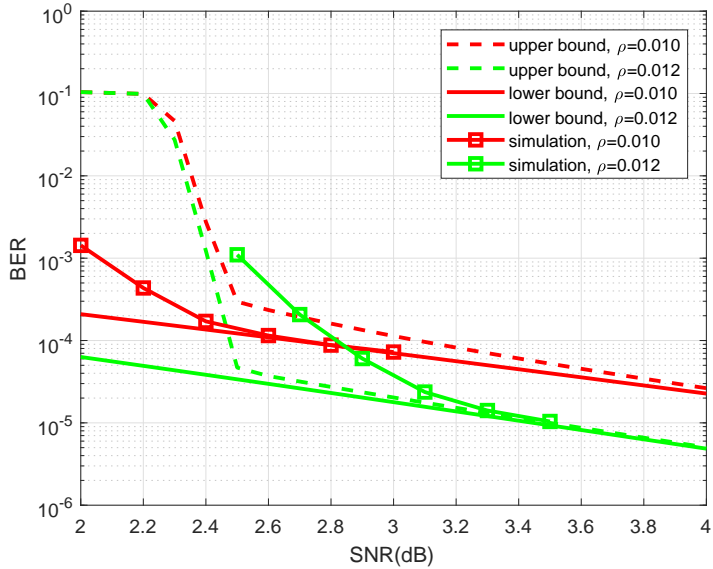}\\
  \caption{The simulated BER performance in Example~\ref{EXEX_BC_fix}. We consider the systematic LDGM code ensemble $\mathscr{C}_s[2048,1024]$ with $\rho=0.010$ and $\rho=0.012$. The corresponding upper bounds and lower bounds are also plotted.}\label{FIG_BC_fix}
\end{figure}
\subsection{Density Evolution Analysis over BECs}
To predict more accurately the performance of iterative BP decoding in the near-capacity region, we turn to density evolution~(DE) analysis for the systematic LDGM ensemble.
For simplicity, we take as an example the BEC with erasure probability $\alpha$.
Let $\varepsilon^{(\ell)}$ and $\eta^{(\ell)}$ be the output erasure probabilities from nodes of type \fbox{=} and \fbox{+} at the $\ell$-th iteration, respectively.
The procedure of DE analysis is initialized by $\eta^{(0)}=1$, since nothing is known at the nodes of \fbox{+} at the beginning of decoding.
Then the erasure probabilities are updated as follows.
\begin{itemize}
\item At a node of type \fbox{=}, the output extrinsic message is an erasure if and only if all other input messages~(including the channel input) are erasures.
 The input from the channel is an erasure with probability $\alpha$, while the input from a check node is an erasure with probability $1-\rho(1-\eta^{(\ell)})$, where $\rho(1-\eta^{(\ell)})$ is the probability that a message is correct and the check node is connected.
 Hence, the probability that a message delivered by a node of type \fbox{=} is an erasure can be calculated as
\begin{equation}\label{EQ_eps}
    \varepsilon^{(\ell)}=\alpha(1-\rho(1-\eta^{(\ell)}))^{n-k-1}.
   \end{equation}
\item At a node of type \fbox{+}, the output extrinsic message is a correct symbol if and only if all other input messages~(including the channel input) are correct.
    The input from the channel is correct with probability $1-\alpha$, while the input from a variable node is correct with probability $1-\rho\varepsilon^{(\ell)}$, where $\rho\varepsilon^{(\ell)}$ is the probability that a message is an erasure and the variable node is connected.
    Hence, the probability that a message delivered by a node of type \fbox{+} is an erasure can be calculated as
    \begin{equation}\label{EQ_eta}
    \eta^{(\ell+1)}=1-(1-\alpha)(1-\rho\varepsilon^{(\ell)})^{k-1}.
    \end{equation}
    \end{itemize}
Obviously, we see that $\eta^{(0)}=1>\eta^{(1)}$, and we can easily prove by induction that $\{\eta^{(\ell)}\}$ is a positive decreasing sequence.
Hence, we can define $\eta^*=\lim_{\ell\rightarrow\infty}\eta^{(\ell)}$, and the estimated erasure probability $\beta$ can be given as
    \begin{equation}
    \beta=\alpha(1-\rho(1-\eta^{*}))^{n-k}.
    \end{equation}
%

\begin{example}\label{EXEX_BC_DE}
Consider the systematic LDGM code ensemble $\mathscr{C}_s[1024,512]$ with $\rho=0.012$.
The DE result and simulated decoding erasure rate are shown in Fig.~\ref{FIG_BC_DE}, where the corresponding upper bounds and lower bounds are also plotted.
We observe that the DE analysis predicts the decoding performance more accurately than the bounds.
However, in the near-capacity region~(around $\alpha=0.5$), a significant gap still exists between the simulation result and the DE result.
This can be explained as below.

Without loss of generality, we focus on the decoding of the first bit $u_0$.
Let $W$ be the number of ones in the first row of matrix $\mathbf{P}$, where the distribution of $W$ is given by
\begin{equation}
\mathrm{Pr}\{W=w\}=\binom{n-k}{w}\rho^w(1-\rho)^{n-k-w}.
\end{equation}
From the view of the receiver, $u_0$ only affects $W+1$ out of $n$ components of the received vector $\bm{y}$.
Equivalently, we say that $u_0$ is transmitted $W+1$ times, among which once~(as an information bit) is over the considered BIOS channel and $W$ times~(as parity-check bits) are also over the considered BIOS channel but with binary interferences from other information bits.
To be more clear, consider a code with the generator matrix
$$\mathbf{G}=\left[\begin{matrix}
1 &0 &0 &1 &0 &1\\
0 &1 &0 &0 &1 &1\\
0 &0 &1 &1 &0 &1\\
\end{matrix}\right],$$
transmitting over BSCs. The message $u_0$ is transmitted three times. One is over BSC and the other two are also over the BSCs but with binary interferences.
The resulting  received symbols are, respectively, $y_0=u_0+e_0$, $y_3=u_0+u_2+e_3$ and $y_5=u_0+u_1+u_2+e_5$, where $e_i$'s are i.i.d. Bernoulli noises and $u_1$, $u_2$ can be viewed as interferences.
When decoding $u_0$, the messages associated with other information bits can be viewed as side information for the interference channels, which are iteratively updated in the BP algorithm\footnote{If all the side information~(all other information bits) were perfectly known at the decoder, the performance of $u_0$ should achieve the lower bound, which is indeed derived by assuming that $u_0$ is transmitted $W+1$ times over the BIOS channel.}.
For DE analysis, these channel side information are assumed to be statistically independent and become more reliable with the iterations.
However, when decoding finite-length codes, they are typically correlated especially when $\rho$ is relatively large.

%
%
\end{example}
\begin{figure}
  \centering
  \includegraphics[width=\figwidth]{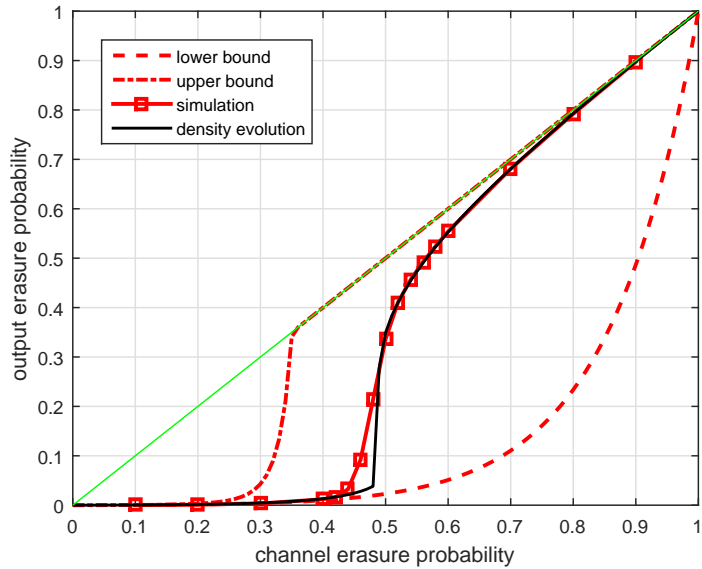}\\
  \caption{The DE result and simulated erasure rate in Example~\ref{EXEX_BC_DE}. We consider the systematic LDGM code ensemble $\mathscr{C}_s[1024,512]$ with $\rho=0.012$. The corresponding upper bounds and lower bounds are also plotted.}\label{FIG_BC_DE}
\end{figure}

\section{Systematic Spatially Coupled LDGM Codes}\label{SEC_4}

As shown in Section~\ref{SEC_3}, the systematic LDGM block code under iterative decoding does not perform well~(say, the performance is even worse than the MAP upper bound) in the low SNR region.
This is mainly caused by the relatively high density of the generator matrix.
To reduce the density but remain the row weights, which dominate the error floors, we turn to the LDPC convolutional codes, also known as spatially coupled LDPC~(SC-LDPC) codes~\cite{Felstrom1999SCLDPC}.
An important feature of SC-LDPC codes is the threshold saturation~\cite{Urbanke2011ThresholdSaturation} that the decoding performance of SC-LDPC codes under BP decoding approaches the MAP decoding performances of uncoupled LDPC block codes. The threshold saturation of SC-LDPC codes implies that the waterfall performance of the systematic LDGM codes, which can also be viewed as LDPC codes, can be improved by spatial coupling. In this section, we present the systematic convolutional LDGM codes~(SC-LDGM\footnote{This acronym can also be interpreted here as ``spatially coupled LDGM" without causing any inconvenience.}) and show that they have the following attractive properties.
\begin{itemize}
  \item They are easily configurable in the sense that any rational code rate can be achieved without complex optimization.
  \item They share the same closed-form lower bounds with the corresponding systematic LDGM codes, which can be used to predict the error floors.
  \item They are iteratively decodable with performance approaching the theoretical limits in the waterfall region and matching with the analytical bounds in the error floor region.
\end{itemize}

\subsection{Encoding Algorithm}
Let $\boldsymbol{u} = (\boldsymbol{u}^{(0)},\cdots,\boldsymbol{u}^{(L-1)})$ be the data to be transmitted, where $\boldsymbol{u}^{(t)}=(u^{(t)}_0,\cdots,u^{(t)}_{k-1})\in\mathbb{F}_2^k$ for $0\leqslant t\leqslant L-1$. The encoding algorithm of the SC-LDGM code with memory $m$ is described in Algorithm~\ref{ALG_ENC_SC}~(see Fig.~\ref{FIG_SCENC} for reference), where $\mathbf{P}_\ell = \{P_{\ell,i,j}\}$ for $0 \leqslant \ell \leqslant m$ is a matrix of size $k\times (n-k)$ generated by a random splitting process. In the random splitting process, each element of $\mathbf{P}$, which is generated according to Definition~\ref{DEF_LDGM_BLOCK_EN}, is sent to $\mathbf{P}_{\ell}$ with probability $1/(m+1)$~(see Fig.~\ref{FIG_SPLIT} for reference). More precisely, the $\mathbf{P}_{\ell}$ can be generated as the following steps.
\begin{itemize}
  \item Generate $\mathbf{P} = \{P_{i,j}\}$ of size $k\times (n-k)$ according to Definition~\ref{DEF_LDGM_BLOCK_EN}.
  \item For each $0\leqslant i \leqslant k$ and $0\leqslant j \leqslant n-k$, draw a random number $s$ independently and uniformly from $\{0,1,\dots,m\}$. Set $P_{\ell,i,j} = P_{i,j}\cdot \delta_{s,\ell}$, for $0 \leqslant \ell \leqslant m$, where $\delta_{s,\ell}$ is the Kronecker delta function.
\end{itemize}

\begin{figure}
  \centering
  \includegraphics[width=\figwidth]{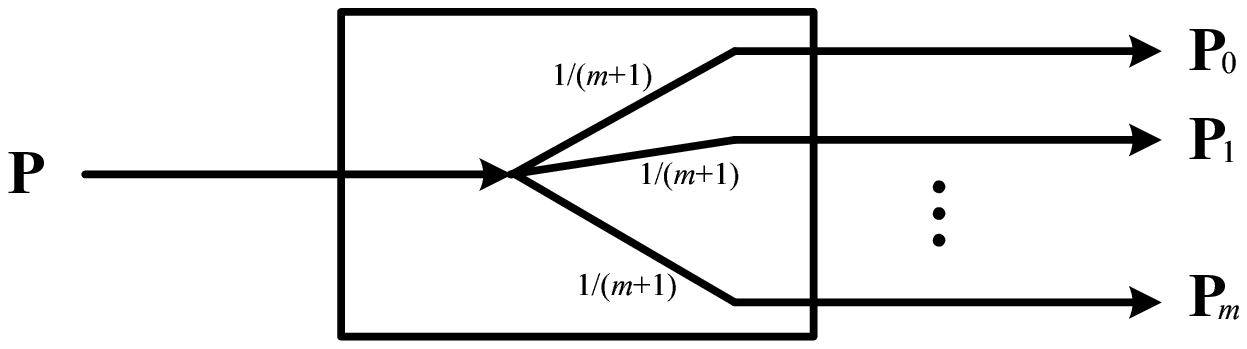}\\
  \caption{Graphical illustrations of the random splitting process.}\label{FIG_SPLIT}
\end{figure}

\begin{figure}
  \centering
  \includegraphics[height=\figwidth,angle=90]{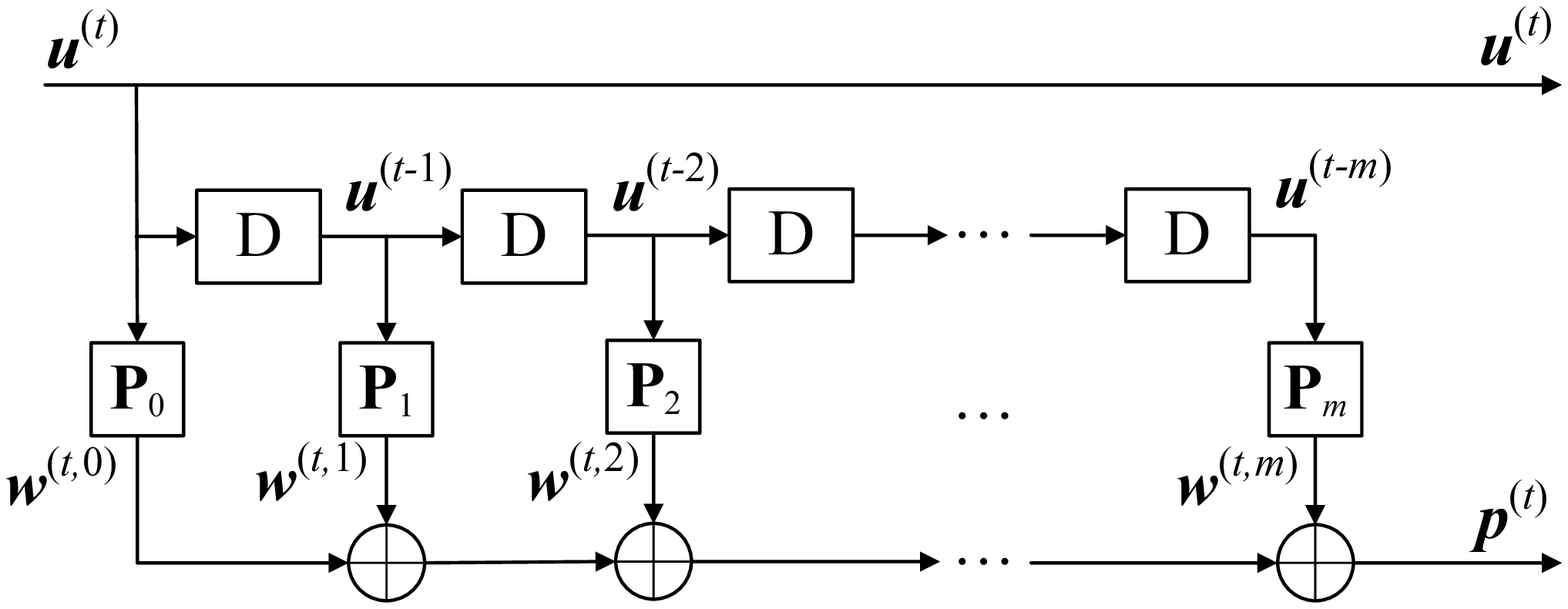}\\
  \caption{Encoding structure of a SC-LDGM code with memory $m$.}\label{FIG_SCENC}
\end{figure}

\begin{algorithm}\caption{Encoding of the SC-LDGM}\label{ALG_ENC_SC}
\begin{itemize}
  \item \textbf{Initialization}: For $t<0$, set $\boldsymbol{u}^{(t)} = \boldsymbol{0} \in \mathbb{F}_2^k$.
  \item \textbf{Iteration}: For $0\leqslant t\leqslant L-1$,
  \begin{itemize}
    \item For $0 \leqslant \ell \leqslant m$, compute $\boldsymbol{w}^{(t,\ell)} = \boldsymbol{u}^{(t)}\mathbf{P}_\ell \in \mathbb{F}_2^{(n-k)}$.
    \item Compute $\boldsymbol{p}^{(t)} = \sum_{\ell=0}^{m}\boldsymbol{w}^{(t,\ell)}$.
    \item Take $\boldsymbol{c}^{(t)} = (\boldsymbol{u}^{(t)},\boldsymbol{p}^{(t)})$ as the sub-frame for transmission at time $t$.
  \end{itemize}
  \item \textbf{Termination}: For $L\leqslant t\leqslant L+m-1$, set $\boldsymbol{u}^{(t)} = \boldsymbol{0} \in \mathbb{F}_2^k$ and compute $\boldsymbol{c}^{(t)}$ following Step \textbf{Iteration}. Note that the information bits~(known as zero bits) should not be transmitted.
\end{itemize}
\end{algorithm}
The total code rate of the SC-LDGM code is $R=\frac{kL}{nL+m(n-k)}$, which is slightly less than that of the systematic LDGM block code. However, the rate loss can be negligible for large $L$.

\subsection{Algebraic Description}
The generator matrix of size $kL\times[nL+m(n-k)]$ of the SC-LDGM code can be written as
\begin{equation}
\mathbf{G} = \left(
 \begin{array}{cccccccccc}
   \mathbf{I} &            &            &            &\mathbf{P}_{0}  & \cdots         & \mathbf{P}_{m} &                &        &\\
              & \mathbf{I} &            &            &                & \mathbf{P}_{0} & \cdots         & \mathbf{P}_{m} &        &\\
              &            & \ddots     &            &                &                &     \ddots     & \ddots         & \ddots &\\
              &            &            & \mathbf{I} &                &                &                &\mathbf{P}_{0}  & \cdots & \mathbf{P}_{m}\\
 \end{array}
\right).
\end{equation}
By the random splitting process, we have $\mathbf{P} = \sum_{\ell=0}^{m}\mathbf{P}_{\ell}$. Therefore, the row weight distribution of the SC-LDGM code ensemble is the same as that of the corresponding systematic LDGM code ensemble, indicating that the SC-LDGM code ensemble shares the same closed-form lower bound with the corresponding systematic LDGM code ensemble.

The SC-LDGM code has a parity-check matrix of banded diagonal form, as is the same case for the SC-LDPC code. So the SC-LDGM code can also be constructed by randomly splitting the parity-check matrix of the LDGM block code. More precisely, The parity-check matrix of size $[(L+m)(n-k)] \times [nL+m(n-k)]$ of the SC-LDGM code can be written as
\begin{equation}
\mathbf{H} = \left(
 \begin{array}{cccccccccc}
   \mathbf{P}_{0}^T &                  &        &                  & \mathbf{I} & & & & &\\
   \vdots           & \mathbf{P}_{0}^T &        &                  & & \mathbf{I} & & & &\\
   \mathbf{P}_{m}^T & \vdots           & \ddots &                  & & & \ddots     & & &\\
                    & \mathbf{P}_{m}^T & \ddots & \mathbf{P}_{0}^T & & & & \ddots     & &\\
                    &                  & \ddots & \vdots           & & & & & \mathbf{I} &\\
                    &                  &        & \mathbf{P}_{m}^T & & & & & & \mathbf{I}\\
 \end{array}
\right).
\end{equation}

\subsection{Decoding Algorithm}
Fig.~\ref{FIG_SCDEC} shows the high-level normal graph of a SC-LDGM code, where an edge represents a sequence of random variables and its associated messages are collectively written in a sequence. Notice that such a simplified representation is just for the convenience of describing the message passing. For message processing, any edge that represents multiple random variables must be treated as multiple separated edges. The high-level normal graph of a SC-LDGM code can be divided into layers, where each layer typically consists of a node of type \fbox{=}, a node of type \fbox{+} and $m+1$ nodes of type \fbox{$\mathbf{P}_{\ell}$}. The node of type \fbox{=} and the node of type \fbox{+} are the same as those discussed in Subsection~\ref{SUBSEC_DEC_BC}. The node of type \fbox{$\mathbf{P}_{\ell}$} can be viewed as a soft-in soft-out decoder of an LDGM block code, which performs Algorithm~\ref{ALG_DEC_BC} to compute the soft outputs~(extrinsic messages). The edge connecting \fbox{$\mathbf{P}_{\ell}$} and \fbox{=} represents the information bits of the LDGM block code, while the edge connecting \fbox{$\mathbf{P}_{\ell}$} and \fbox{+} represents the parity-check bits.

Assume that $\boldsymbol{c}^{(t)}$ is modulated using BPSK and transmitted over an AWGN channel, resulting in a received vector $\boldsymbol{y}^{(t)}$. We consider an iterative sliding window decoding algorithm with a fixed decoding delay $d\geqslant0$. The iterative sliding-window algorithm with decoding delay $d$ works over a subgraph consisting of $d + 1$ consecutive layers. The schedule is described in Algorithm~\ref{ALG_DEC_SC}.

\begin{figure}
  \centering
  \includegraphics[height=\figwidth,angle=90]{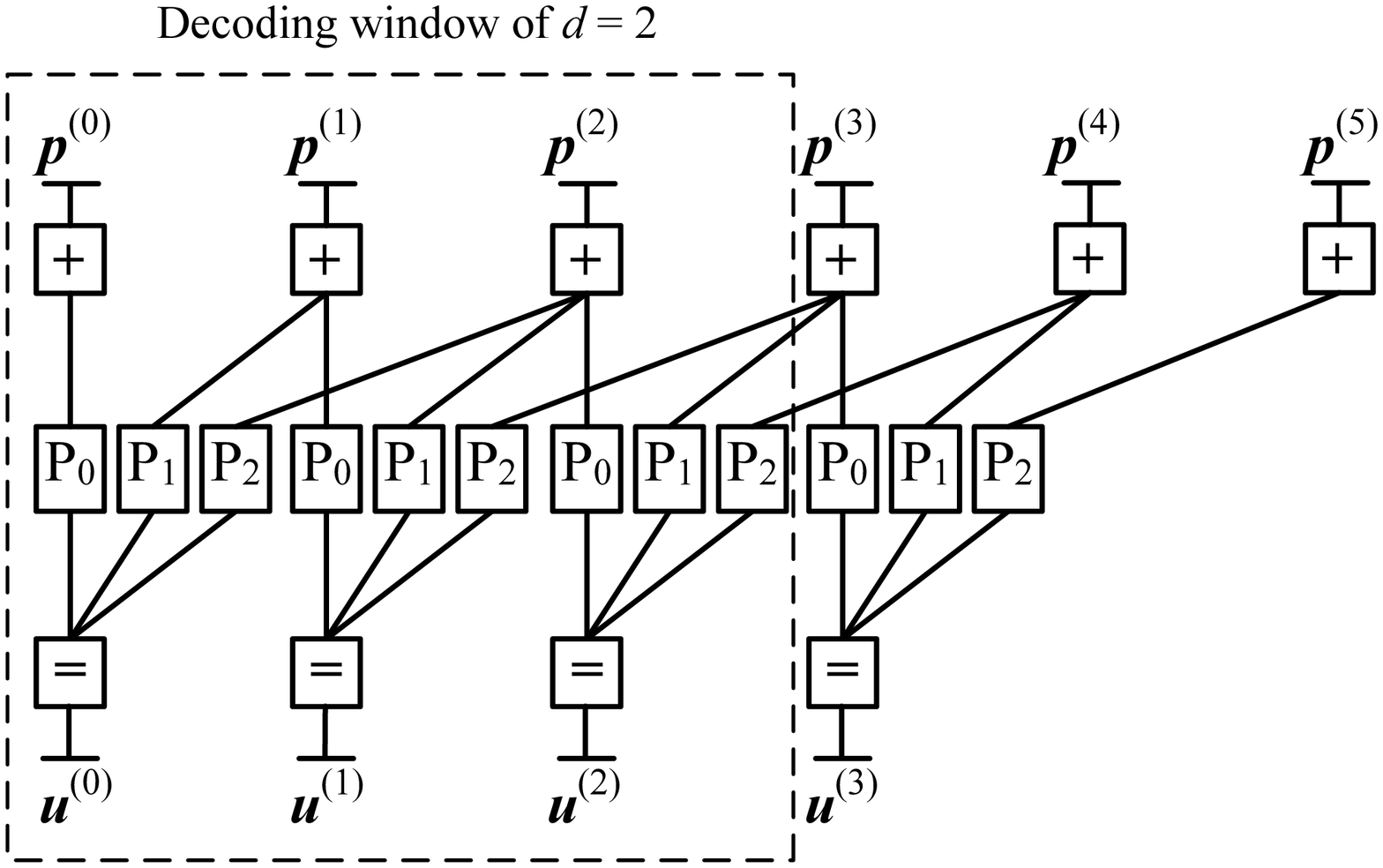}\\
  \caption{The normal graph of a SC-LDGM code with $L = 4$ and $m = 2$. The decoding window of $d=2$ is also plotted.}\label{FIG_SCDEC}
\end{figure}

\begin{algorithm}\caption{Iterative Sliding Window Decoding of the SC-LDGM}\label{ALG_DEC_SC}
\begin{itemize}
  \item \textbf{Global initialization}: Set a maximum global iteration $J_{\rm max}$. For $0 \leqslant t \leqslant d-1$, considering only the channel constraint, compute the \emph{a posteriori} probabilities associated with $\boldsymbol{c}^{(t)}$ from the received vector $\boldsymbol{y}^{(t)}$. All messages over the other edges within and connecting to the decoding window are initialized as uniformly distributed variables.
  \item \textbf{Sliding window decoding}: For $t=0,1,\cdots,L-1$,
  \begin{enumerate}
    \item Local initialization: If $t+d \leqslant L+m-1$, compute the \emph{a posteriori} probabilities from the received vector $\boldsymbol{y}^{(t+d)}$ and all messages over other edges within and connecting to the $(t + d)$-th layer are initialized as uniformly distributed variables.
    \item Iteration: For $j=1,2,\cdots,J_{\rm max}$,
    \begin{enumerate}
      \item Forward recursion: For $i=0,1,\cdots,d$, the $(t+i)$-th layer performs a message passing algorithm scheduled as
      \begin{displaymath}
      \fbox{+}\rightarrow\fbox{$\mathbf{P}_{0}$}\rightarrow\cdots\rightarrow\fbox{$\mathbf{P}_{m}$}\rightarrow\fbox{=}\rightarrow\fbox{$\mathbf{P}_{0}$}\rightarrow\cdots\rightarrow\fbox{$\mathbf{P}_{m}$}
      \end{displaymath}
      \item Backward recursion: For $i=d,d-1,\cdots,0$, the $(t+i)$-th layer performs a message passing algorithm scheduled as
      \begin{displaymath}
      \fbox{=}\rightarrow\fbox{$\mathbf{P}_{0}$}\rightarrow\cdots\rightarrow\fbox{$\mathbf{P}_{m}$}\rightarrow\fbox{+}\rightarrow\fbox{$\mathbf{P}_{0}$}\rightarrow\cdots\rightarrow\fbox{$\mathbf{P}_{m}$}
      \end{displaymath}
      \item Decision: Make decision on $\boldsymbol{u}^{(t)}$, resulting in $\hat{\boldsymbol{u}}^{(t)}$. If the entropy-based stopping criterion~\cite{Ma2015Block} are satisfied, output $\hat{\boldsymbol{u}}^{(t)}$ and exit the iteration.
    \end{enumerate}
    \item Cancelation: Remove the effect of $\hat{\boldsymbol{u}}^{(t)}$ on $\boldsymbol{y}^{(t+1)},\cdots,\boldsymbol{y}^{(t+m)}$.
  \end{enumerate}
\end{itemize}
\end{algorithm}

\subsection{Numerical Results}
We first consider the erasure channel, in which the performance can be estimated by DE analysis.

\begin{example3}
In order to compare with the systematic LDGM code, we consider SC-LDGM codes with tail-bitting, which has the same node degree distribution as the LDGM block code.
Hence, these codes share the same DE result and the lower bound, which are shown in Fig.~\ref{FIG_SC_DE}.
We can observe that the decoding performance improves slightly with the increase of memory $m$ and the performance of SC-LDGM codes approach the DE result, which is as expected and confirms our analysis.
\end{example3}

Actually, with spatial coupling, the side information with respect to those interference channels are collected from, on average, $(k-1)\rho$ out of $(m+1)k-1$ rather than $k-1$  information bits.
Hence we expect that the performance of SC-LDGM codes can be predicted more accurately by DE analysis when $m$ is relatively large.



\begin{figure}
  \centering
  \includegraphics[width=\figwidth]{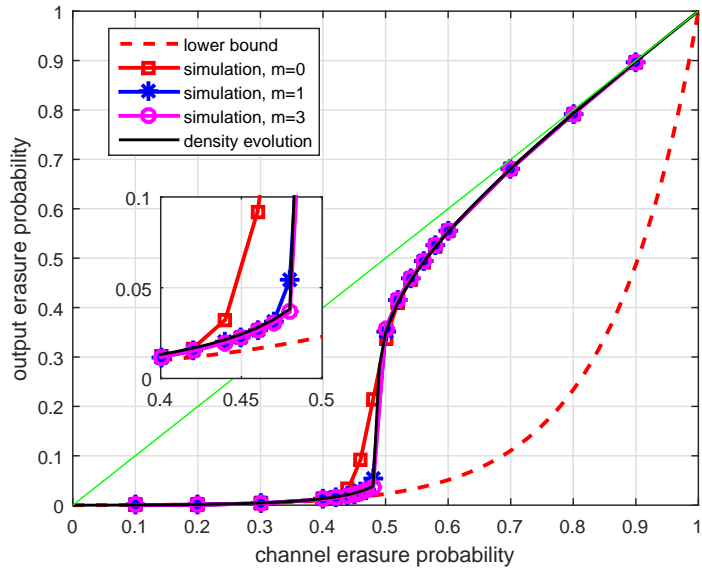}\\
  \caption{The decoding erasure rate performance of the SC-LDGM codes in Example~\ref{EXEX_BC_DE}~(Cont'd). We consider SC-LDGM codes with memories $m=0,1,3$ corresponding to a systematic LDGM code $\mathscr{C}_s[1024,512]$ with $\rho=0.012$. The data block length is set as $L=150$. The DE result, upper bound and the lower bound are also plotted.}\label{FIG_SC_DE}
\end{figure}

In the remainder of this section, we focus on the BPSK-AWGN channels. The iterative sliding window decoding algorithm is performed with a maximum iteration number of $J_{\rm max}=18$ and a decoding window $d=2m$. The threshold for the entropy stopping criterion is set as $\epsilon = 10^{-5}$.

\begin{example}\label{EXEX_SC_m}
Consider SC-LDGM codes with memories $m=0,1,3,6$ corresponding to a systematic LDGM code $\mathscr{C}_s[2048,1024]$ with $\rho=0.012$. The data block length is set as $L=150$. The BER performance is shown in Fig.~\ref{FIG_SC_m}, where the upper bound of the LDGM block code~($m=0$) and the lower bound of all codes are also plotted. We can observe that
\begin{itemize}
  \item The simulated BER performance curve matches the lower bound well in the high SNR region, implying that the iterative sliding window decoding algorithm is near optimal~(with respect to the MAP decoding algorithm) in the high SNR region.
  \item The waterfall performance of the SC-LDGM code is better than that of the corresponding LDGM block code~($m=0$). The waterfall performance improves with increasing encoding memory $m$.
\end{itemize}

\end{example}
\begin{figure}
  \centering
  \includegraphics[width=\figwidth]{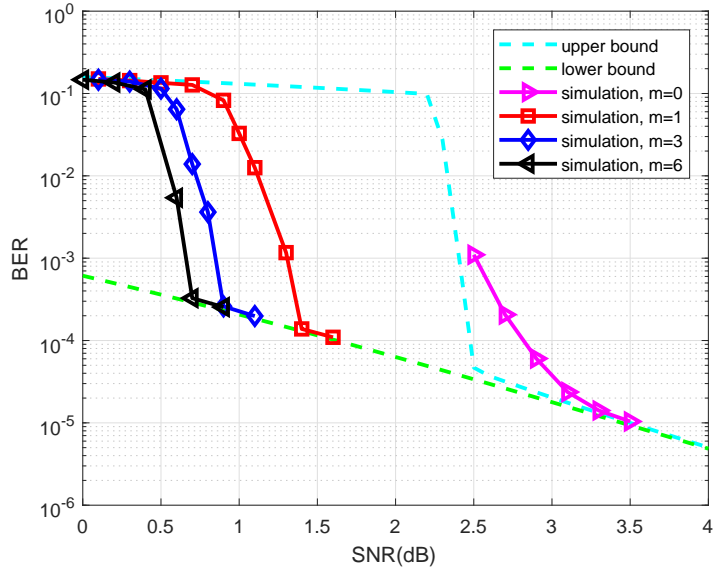}\\
  \caption{The BER performance of the SC-LDGM codes in Example~\ref{EXEX_SC_m}. We consider SC-LDGM codes with memories $m=0,1,3,6$ corresponding to a systematic LDGM code $\mathscr{C}_s[2048,1024]$ with $\rho=0.012$. The data block length is set as $L=150$. The upper bound of the LDGM block code~($m=0$) and the lower bound of all codes are also plotted.}\label{FIG_SC_m}
\end{figure}

\begin{example}\label{EXEX_SC_rho}
Consider the SC-LDGM codes corresponding to the systematic LDGM codes $\mathscr{C}_s[2048,1024]$. The parameters $\rho$ and $m$ are specified in the legends. All codes are terminated properly such that the total rates are $R=0.49$. The BER performance is shown in Fig.~\ref{FIG_SC_rho}, where the corresponding lower bounds are also plotted. We can observe that the error floor can be lowered by increasing $\rho$. However, for a fixed memory $m$, under the sub-optimal iterative decoding, the waterfall performance with a large $\rho$ is typically worse than that with a small $\rho$. Therefore, the memory $m$ of the SC-LDGM code with a large $\rho$ should be set large, implying that the decoding window $d$~(hence the decoding latency) will be large.
\end{example}
\begin{figure}
  \centering
  \includegraphics[width=\figwidth]{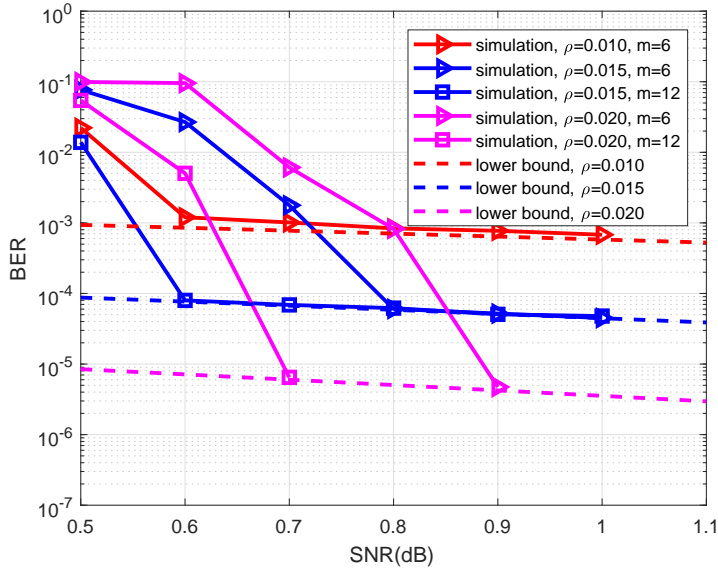}\\
  \caption{The BER performance of the SC-LDGM codes in Example~\ref{EXEX_SC_rho}. We consider the SC-LDGM codes corresponding to the systematic LDGM codes $\mathscr{C}_s[2048,1024]$. The parameters $\rho$ and $m$ are specified in the legends. All codes are terminated properly such that the total rates are $R=0.49$. The corresponding lower bounds are also plotted.}\label{FIG_SC_rho}
\end{figure}

\begin{example}\label{EXEX_SC_rate}
Consider the SC-LDGM codes with memory $m=6$. The information subsequence length is $k=1024$ and the data block length is $L=300$. The parameters $\rho$ and the total code rate are specified in the legends. The BER performance is shown in Fig.~\ref{FIG_SC_rate}, where the corresponding lower bounds and shannon limits under BPSK constraint are also plotted. We can observe that, as the SNR increases, the performance curves of the SC-LDGM codes drop down to the respective lower bounds for all considered code rates. We see that the SC-LDGM code performs about $0.7~{\rm dB}$ away from the respective Shannon limits at the BER of $10^{-4}$ for all SC-LDGM codes given in Fig.~\ref{FIG_SC_rate}.
\end{example}
\begin{figure}
  \centering
  \includegraphics[width=\figwidth]{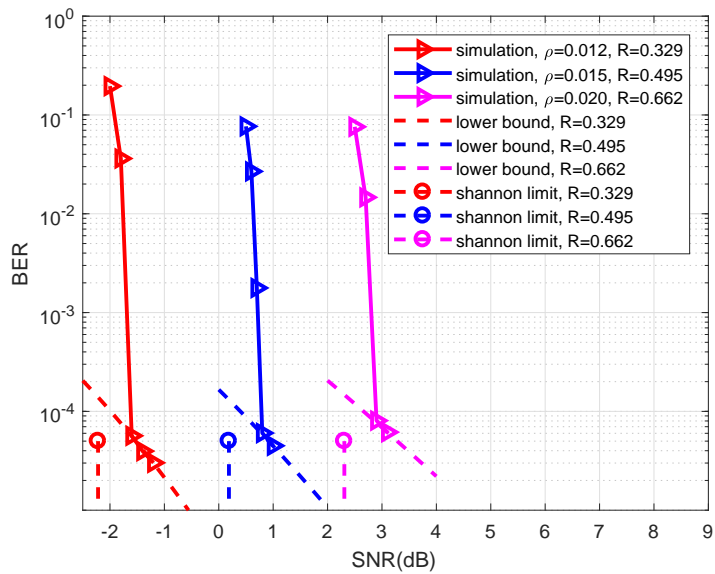}\\
  \caption{The BER performance of the SC-LDGM codes in Example~\ref{EXEX_SC_rate}. We consider the SC-LDGM codes with memory $m=6$. The information subsequence length is $k=1024$ and the data block length is $L=300$. The parameters $\rho$ and the total code rate are specified in the legends. The corresponding lower bounds and shannon limits under BPSK constraint are also plotted.}\label{FIG_SC_rate}
\end{figure}

\section{Discussion and Conclusion}\label{SEC_5}

In this paper, we have proposed an LDGM code ensemble, which is defined by the Bernoulli process. For asymptotic performance analysis, we have proved that the proposed ensemble is capacity-achieving over BIOS memoryless channels. The proof technique is different from existing ones whereby the performance criterion is BER instead of FER. For finite length performance analysis, an upper bound and a lower bound are presented, both of which are tight in the high SNR region and helpful to predict the error floor and to evaluate the near-optimality of the iterative decoding algorithm.

Practically, we focus on the performance in the error floor and waterfall regions. To lower down the error floor, we can simply increase the coding length or the probability $\rho$ in the generator matrix, as shown in Example~1 and Example~2.
As mentioned in~\cite{Spielman1996LDPC}, the LDGM codes can be treated as the ``error reduction'' codes, which can be used as the inner code for the concatenated coding system.
We show by DE analysis over BECs that the iterative decoding performance can be improved by a sparser generator matrix.
 Hence we employ the spatial coupling technique to improve the waterfall performance, which has been proved to be effective for LDPC codes, and proposed the SC-LDGM codes.
  The main advantage of the presented SC-LDGM codes is their flexible construction and predicable performance. Numerical results showed that, under iterative BP decoding algorithm, the SC-LDGM codes perform about $0.7~{\rm dB}$ away from the Shannon limits for various code rates.

%
%

\ifCLASSOPTIONcaptionsoff
  \newpage
\fi

\bibliographystyle{IEEEtran}
\bibliography{IEEEabrv,bibfile}

\begin{thebibliography}{10}
\providecommand{\url}[1]{#1}
\csname url@samestyle\endcsname
\providecommand{\newblock}{\relax}
\providecommand{\bibinfo}[2]{#2}
\providecommand{\BIBentrySTDinterwordspacing}{\spaceskip=0pt\relax}
\providecommand{\BIBentryALTinterwordstretchfactor}{4}
\providecommand{\BIBentryALTinterwordspacing}{\spaceskip=\fontdimen2\font plus
\BIBentryALTinterwordstretchfactor\fontdimen3\font minus
  \fontdimen4\font\relax}
\providecommand{\BIBforeignlanguage}[2]{{%
\expandafter\ifx\csname l@#1\endcsname\relax
\typeout{** WARNING: IEEEtran.bst: No hyphenation pattern has been}%
\typeout{** loaded for the language `#1'. Using the pattern for}%
\typeout{** the default language instead.}%
\else
\language=\csname l@#1\endcsname
\fi
#2}}
\providecommand{\BIBdecl}{\relax}
\BIBdecl

\bibitem{Shannon1948Theory}
C.~Shannon, ``A mathematical theory of communication,'' \emph{Bell Syst. Tech.
  J.}, vol.~27, no.~3, pp. 379--423, July 1948.

\bibitem{Elias1955Coding}
P.~Elias, ``Coding for noisy channels,'' \emph{IRE Conv. Rec.}, vol.~4, pp.
  37--46, 1955.

\bibitem{Gallager68}
R.~Gallager, \emph{Information Theory and Reliable Communication}.\hskip 1em
  plus 0.5em minus 0.4em\relax New York, NY: John Wiley and Sons, Inc., 1968.

\bibitem{Shulman1999RandomCoding}
N.~Shulman and M.~Feder, ``Random coding techniques for nonrandom codes,''
  \emph{IEEE Trans. Inf. Theory}, vol.~45, no.~6, pp. 2101--2104, Sept. 1999.

\bibitem{Gallager1962ldpc}
R.~Gallager, \emph{Low-Density Parity-check Codes}.\hskip 1em plus 0.5em minus
  0.4em\relax Cambridge, MA: MIT Press, 1963.

\bibitem{MacKay1997LDPC}
D.~MacKay and R.~Neal, ``Near {Shannon} limit performance of low density parity
  check codes,'' \emph{IET Electron. Lett.}, vol.~33, no.~6, pp. 457--458, Mar.
  1997.

\bibitem{Spielman1996LDPC}
D.~Spielman, ``Linear-time encodable and decodable error-correcting codes,''
  \emph{IEEE Trans. Inf. Theory}, vol.~42, no.~6, pp. 1723--1731, Nov. 1996.

\bibitem{Urbanke2001Irregular}
T.~Richardson, M.~Shokrollahi, and R.~Urbanke, ``Design of capacity-approaching
  irregular low-density parity-check codes,'' \emph{IEEE Trans. Inf. Theory},
  vol.~47, no.~2, pp. 619--637, Feb. 2001.

\bibitem{Urbanke2001LDPC45}
S.~Chung, G.~Forney, T.~Richardson, and R.~Urbanke, ``On the design of
  low-density parity-check codes within 0.0045 {dB} of the {Shannon} limit,''
  \emph{IEEE Commun. Lett.}, vol.~5, no.~2, pp. 58--60, Feb. 2001.

\bibitem{Shokrollahi2006Raptor}
A.~Shokrollahi, ``Raptor codes,'' \emph{IEEE Trans. Inf. Theory}, vol.~52,
  no.~6, pp. 2551--2567, Jun. 2006.

\bibitem{Luby2002LT}
M.~Luby, ``{LT} codes,'' in \emph{Annu. IEEE Symp. Found. Comput. Sci.},
  Vancouver, Canada, Nov. 2002, pp. 271--280.

\bibitem{Kakhaki2012Sparse}
A.~Kakhaki, H.~Abadi, P.~Pad, H.~Saeedi, F.~Marvasti, and K.~Alishahi,
  ``Capacity achieving linear codes with random binary sparse generating
  matrices over the binary symmetric channel,'' in \emph{Int. Symp. Inf.
  Theory}, Cambridge, USA, July 2012, pp. 621--625.

\bibitem{Kumar2014ThresholdSaturation}
S.~Kumar, A.~Young, N.~Macris, and H.~Pfister, ``Threshold saturation for
  spatially coupled {LDPC} and {LDGM} codes on {BMS} channels,'' \emph{IEEE
  Trans. Inf. Theory}, vol.~60, no.~12, pp. 7389--7415, Dec. 2014.

\bibitem{Urbanke2011ThresholdSaturation}
S.~Kudekar, T.~Richardson, and R.~Urbanke, ``Threshold saturation via spatial
  coupling: Why convolutional ldpc ensembles perform so well over the {BEC},''
  \emph{IEEE Trans. Inf. Theory}, vol.~57, no.~2, pp. 803--834, Feb. 2011.

\bibitem{Kudekar2010}
S.~Kudekar, C.~Measson, T.~Richardson, and R.~Urbanke, ``Threshold saturation
  on {BMS} channels via spatial coupling,'' in \emph{Int. Symp. Turbo Codes
  Iterative Inf. Process.}, Brest, France, Sept. 2010, pp. 309--313.

\bibitem{Kudekar2013}
S.~Kudekar, T.~Richardson, and R.~Urbanke, ``Spatially coupled ensembles
  universally achieve capacity under belief propagation,'' \emph{IEEE Trans.
  Inf. Theory}, vol.~59, no.~12, pp. 7761--7813, Dec. 2013.

\bibitem{Ma2016Coding}
X.~Ma, ``Coding theorem for systematic low density generator matrix codes,'' in
  \emph{Int. Symp. Turbo Codes Iterative Inf. Process.}, Brest, France, Sept.
  2016, pp. 11--15.

\bibitem{Lin2018Coding}
W.~Lin, S.~Cai, B.~Wei, and X.~Ma, ``Coding theorem for systematic {LDGM} codes
  under list decoding,'' in \emph{Inf. Theory Workshop}, Guangzhou, China, Nov.
  2018, pp. 1--5.

\bibitem{MacKay1999Irregular}
D.~MacKay, S.~Wilson, and M.~Davey, ``Comparison of constructions of irregular
  {Gallager} codes,'' \emph{IEEE Trans. Commun.}, vol.~47, no.~10, pp.
  1449--1454, Oct. 1999.

\bibitem{Urbanke2001Irregular2}
T.~Richardson and R.~Urbanke, ``The capacity of low-density parity-check codes
  under message-passing decoding,'' \emph{IEEE Trans. Inf. Theory}, vol.~47,
  no.~2, pp. 599--618, Feb. 2001.

\bibitem{Ma2017Systematic}
X.~Ma, K.~Huang, and B.~Bai, ``Systematic block {Markov} superposition
  transmission of repetition codes,'' \emph{IEEE Trans. Inf. Theory}, vol.~64,
  no.~3, pp. 1604--1620, Mar. 2018.

\bibitem{Benedetto1996IRWEF}
S.~Benedetto and G.~Montorsi, ``Unveiling turbo codes: Some results on parallel
  concatenated coding schemes,'' \emph{IEEE Trans. Inf. Theory}, vol.~42,
  no.~2, pp. 409--428, Mar. 1996.

\bibitem{Arikan2009ChannelPolarization}
E.~{Arikan}, ``Channel polarization: A method for constructing
  capacity-achieving codes for symmetric binary-input memoryless channels,''
  \emph{IEEE Trans. Inf. Theory}, vol.~55, no.~7, pp. 3051--3073, July. 2009.

\bibitem{Forney2001Graph}
G.~Forney, ``Codes on graphs: Normal realizations,'' \emph{IEEE Trans. Inf.
  Theory}, vol.~47, no.~2, pp. 520--548, Feb. 2001.

\bibitem{Felstrom1999SCLDPC}
A.~Felstrom and K.~Zigangirov, ``Time-varying periodic convolutional codes with
  low-density parity-check matrix,'' \emph{IEEE Trans. Inf. Theory}, vol.~45,
  no.~6, pp. 2181--2191, Sept. 1999.

\bibitem{Ma2015Block}
X.~Ma, C.~Liang, K.~Huang, and Q.~Zhuang, ``Block {Markov} superposition
  transmission: Construction of big convolutional codes from short codes,''
  \emph{IEEE Trans. Inf. Theory}, vol.~61, no.~6, pp. 3150--3163, Jun. 2015.

\end{thebibliography}
\end{document}